\documentclass[a4letter, 10pt, journal]{IEEEtran}\usepackage{amsmath,amssymb,textcomp} 
\usepackage[OT1]{fontenc}
\usepackage{graphicx,mathrsfs}
\usepackage{dblfloatfix}
\usepackage{mathtools}
\usepackage{amsfonts}
\usepackage{bbm}
\usepackage{dsfont}
\usepackage{xcolor}

\newtheorem{theorem}{Theorem}

\newtheorem{corollary}[theorem]{Corollary}

\newtheorem{definition}[theorem]{Definition}
\newtheorem{example}[theorem]{Example}

\newtheorem{lemma}[theorem]{Lemma}

\newtheorem{remark}[theorem]{Remark}

\newenvironment{proof}[1][Proof]{\textbf{#1.} }{\ \rule{0.5em}{0.5em}}

\begin{document}

\title{A Small-Gain Theorem for Discrete-Time Convergent Systems and Its Applications}

\author{Jiayin Chen and Hendra I. Nurdin \thanks{J. Chen and H. I. Nurdin are with the School of Electrical Engineering and Telecommunications,  UNSW Australia, Sydney NSW 2052, Australia (\texttt{email: jiayin.chen@unsw.edu.au, h.nurdin@unsw.edu.au})}}

\maketitle

\begin{abstract} 
Convergent, contractive or incremental stability properties of nonlinear systems have attracted interest for control tasks such as observer design, output regulation and synchronization. The convergence property plays a central role in the neuromorphic (brain-inspired) computing of reservoir computing, which seeks to harness the information processing capability of nonlinear systems. This paper presents a small-gain theorem for discrete-time output-feedback interconnected systems to be uniformly input-to-output convergent (UIOC) with outputs  converging to a bounded reference output uniquely determined by the input. A small-gain theorem for interconnected time-varying discrete-time uniform input-to-output stable systems that could be of separate interest is also presented as an intermediate result. Applications of the UIOC small-gain theorem are illustrated in the design of observer-based controllers and interconnected nonlinear classical and quantum dynamical systems (as reservoir computers) for black-box system identification.\end{abstract}

\begin{IEEEkeywords}
Convergent dynamics; Reservoir computing; Small-gain; Input-to-output stability.
\end{IEEEkeywords}

\section{Introduction}
\label{sec:intro} 
\IEEEPARstart{C}{onvergence} notions, such as incremental stability \cite{angeli2002lyapunov}, convergent dynamics \cite{pavlov2005convergent,pavlov2008convergent} and contracting dynamics \cite{lohmiller1998contraction}, impose that all solutions must ``forget'' their initial conditions and converge to each other asymptotically; also see \cite{tran2018convergence} for a survey of such convergence properties in the discrete-time setting. Such properties have found applications in observer design \cite{angeli2002lyapunov,lohmiller1998contraction}, output regulation \cite{pavlov2011steady,pavlov2006uniform} and synchronization \cite{angeli2002lyapunov,pham2007stable}.

In an independent development, efforts to go beyond the von Neumann computing architecture to imitate human capabilities in tasks that are difficult or energetically expensive on conventional digital computers, have led to the pursuit of neuromorphic (brain-inspired) computing paradigms implemented on physical hardware. Central to this is to exploit high-dimensional nonlinear dynamical systems for processing of time-varying input signals. An emerging neuromorphic paradigm is reservoir computing (RC) \cite{tanaka2019recent,nakajima2020reservoir}, which uses a fixed but otherwise almost arbitrary dynamical system, the so-called “reservoir”, to map inputs into its state-space. This paper is interested in discrete-time reservoir computers (also abbreviated as RCs) that perform causal nonlinear operations on input sequences to produce output sequences. Only a simple linear regression algorithm is required to optimize the parameters of a readout function to approximate target output sequences. In the RC paradigm, convergence (referred to as the echo-state property in the RC literature) ensures that the reservoir outputs are asymptotically independent of its initial condition, and the reservoir induces an input-output (I/O) map to approximate a target I/O map \cite{grigoryeva2018universal}. 

A prominent example of RCs are the echo state networks (ESNs), which have demonstrated  remarkable ability to predict chaotic time series \cite{jaeger2004harnessing,pathak2018model}. There is substantial interest in hardware realizations of RC for fast processing using less memory and energy, as opposed to software-based implementation on digital computers. For instance, a photonic RC achieved high-speed speech classification (million words per second) with low error rates \cite{larger2017high} and a FPGA-based RC reached an 160 MHz rate for time-series prediction \cite{canaday2018rapid}; see \cite{zhou2020reservoir,galan2020tropical,soriano2019optoelectronic} for other recent RC hardware implementations. Its extremely efficient training makes RC suitable for applications such as edge computing, in which information processing is incorporated into decentralized sensors (the `edges') to reduce computation and transmission overhead \cite{nakajima2020physical}.  For more recent developments, see \cite{tanaka2019recent,nakajima2020reservoir,nakajima2020physical}. 

Recent years have also seen the advent of noisy near-term intermediate scale quantum (NISQ) computers \cite{Preskill18}, which are not equipped with quantum error correction and can only perform a limited form of quantum computing, made available via cloud-based access from companies such as IBM \cite{IBMQ}. This has led to the proposal of  RCs that exploit nonlinear quantum dynamics (referred to as QRC) \cite{fujii2017harnessing,CN19}. Control-oriented applications of QRCs were put forward in \cite{JNY19} and Ref.~\cite{chen2020temporal} further develops a QRC scheme that can be implemented on NISQ computers and demonstrates a proof-of-principle of QRC on cloud-based IBM superconducting quantum computers \cite{IBMQ}.  This work provides further theoretical support for the application of RC and QRC for black-box system identification of discrete-time nonlinear systems, with examples developed in Section~\ref{sec:applications}. Since QRCs can have a high-dimensional underlying Hilbert space, it may be advantageous for identifying systems with a high-dimensional state-space.

This paper focuses on uniformly convergent dynamics or the uniform convergence (UC) property. For output-feedback interconnected systems, we introduce the uniform output convergence (UOC) and the uniform input-to-output convergence (UIOC) properties. Roughly speaking, a UOC system has a unique reference state solution with its reference output defined and bounded both backwards and forwards in time. All other outputs asymptotically converge to the reference output, independent of their initial conditions. The UIOC property adapts the uniform input-to-state convergence (UISC) property \cite{pavlov2005convergent} to output-feedback interconnected systems. A UIOC system is UOC and the perturbation in its reference outputs are asymptotically bounded by a nonlinear gain on the input perturbation. We present a small-gain theorem for output-feedback interconnected systems to be UIOC, and that the closed-loop system induces a well-posed I/O map in the sense of \cite{shamma1993fading}.

Our UIOC small-gain theorem is based on a small-gain theorem for time-varying discrete-time systems in the uniform input-to-output stability (UIOS) framework, also presented herein. This latter small-gain result is used to establish the UIOC small-gain theorem for interconnected time-invariant UOC systems, the time-variance arising through a change of variables. Small-gain criteria for time-varying continuous-time interconnected systems in the UIOS framework have been established \cite{jiang1997small,jiang1994small,chen2005simplified,zhu2006small}. Ref.~\cite{sontag2002small} develops a generalized small-gain theorem that  recover the previous results for specific interconnections. 
Small-gain criteria for time-invariant discrete-time systems can be found in \cite{jiang2004nonlinear,jiang2001input}. Here, we adopt the techniques in \cite{sontag2002small} to establish a UIOS small-gain theorem for output-feedback interconnected time-varying discrete-time systems. Ref.~\cite{chen2005simplified} is based on \cite[Lemma~3, Prop. 2.5]{Ylin}, which also concerns continuous-time systems, and is not immediately applicable to our setting. Furthermore, \cite{jiang2004nonlinear,jiang2001input} do not carry over to time-varying systems \cite{jiang2005remarks}. Therefore, we establish a link to bridge our setting with the continuous-time results of \cite{sontag2002small}.

Finally, we apply our results to observer-based controller design for globally Lipschitz systems \cite{ibrir2005observer,ibrir2008novel} and to design RC parameters for black-box system identification of discrete-time nonlinear systems solely based on I/O data collected from a system. Example nonlinear models for black-box system identification include NARMAX \cite{johansen1993constructing}, the Volterra series \cite{boyd1985fading} and block-oriented models \cite{giri2010block,schoukens2017identification}. The use of closed-loop structures, such as in the Wiener-Hammerstein feedback model, is motivated by modeling systems that exhibit nonlinear feedback behavior \cite{schoukens2017identification}. Here we introduce interconnected ESNs and QRCs as models with closed-loop structures. The interconnected ESNs and QRCs dynamics are arbitrary but fixed at the onset, as long as they satisfy the UIOC small-gain theorem. Only an RC output function is optimized via ordinary least squares to fit the data, making the RC approach computationally efficient.  We illustrate numerically the efficacy of this approach to model a feedback-controlled nonlinear system.

\textbf{Notation.} $\| \cdot \|$ is the Euclidean norm. $\mathbb{Z}$ $(\mathbb{R})$ and $\mathbb{Z}_{\geq k_0}$ ($ \mathbb{R}_{\geq k_0}$) denote, respectively, integers (reals) and integers (reals) larger than equal to $k_0 \in \mathbb{Z}$. $x \in \mathbb{R}^n$ is a column vector and $x^\top$ its transpose. For $x_j \in \mathbb{R}^{n_j} (j=1, 2)$, $(x_1, x_2)\in \mathbb{R}^{n_1 + n_2}$ is their concatenation into a column vector. For a sequence $z$ on $\mathbb{Z}$, $\| z_{[k_0, k]} \| \coloneqq \sup_{k_0 \leq j \leq k} \| z(j) \|$ for any $k_0 \in \mathbb{Z}$ and $k \geq k_0$. $l_n^{\infty}$ denotes the set of bounded sequences on $\mathbb{Z}$, i.e., $z \in l^{\infty}_n$ if $z(k) \in \mathbb{R}^n$ and $ \| z \|_{\infty} \coloneqq \sup_{k \in \mathbb{Z}}\|z(k)\| < \infty$. For any sequences $z_1, z_2$ on $\mathbb{Z}$, $z=(z_1, z_2)$ is given by $z(k) = (z_1(k), z_2(k))$  $\forall k \in \mathbb{Z}$. Function composition is denoted by $\circ$. We use standard comparison function classes $\mathcal{K}, \mathcal{K_{\infty}}$ and $\mathcal{KL}$ \cite{jiang1994small}\footnote{A function $\gamma: \mathbb{R}_{\geq 0} \rightarrow \mathbb{R}_{\geq 0}$ is a $\mathcal{K}$ function if it is continuous, strictly increasing and $\gamma(0) = 0$. A $\mathcal{K}$ function is a $\mathcal{K}_{\infty}$ function if it is unbounded. A $\mathcal{KL}$ function $\beta: \mathbb{R}_{\geq 0} \times \mathbb{R}_{\geq 0} \rightarrow \mathbb{R}_{\geq 0}$ is $\mathcal{K}$ in the first argument and decreasing (i.e., non-increasing) in the second argument, with $\lim_{t \rightarrow \infty}\beta(s, t) = 0$ for all $s \in \mathbb{R}_{\geq 0}$. As in \cite{jiang1994small,lin1996smooth}, we do not require a $\mathcal{KL}$ function to be continuous or strictly decreasing in the second argument.}.

\section{Stability concepts}
\label{sec:stability-concepts}
This section defines the uniform output convergence (UOC) (Def.~\ref{defn:uoc}) and the uniform input-to-output convergence (UIOC) (Def.~\ref{defn:uioc}) properties. See Table~\ref{table1:acronym} for a summary of relevant stability definitions. We first set some preliminaries.

For $k \in \mathbb{Z}$, consider a time-varying discrete-time system, 
\begin{equation} \label{eq:DT-sys}
\begin{cases}
x(k+1) = f(k, x(k), u(k)), \\
\hspace{1.8em} y(k) = h(k, x(k), u(k)),
\end{cases}
\end{equation}
where $x(k) \in \mathbb{R}^{n_x}$ is the state, $u(k) \in \mathbb{R}^{n_u}$ is the input and $y(k) \in \mathbb{R}^{n_y}$ is the output. We assume that $u \in l^{\infty}_{n_u}$ and for each $k \in \mathbb{Z}$, $\| f(k, x(k), u(k)) \| < \infty$ and $\| h(k, x(k), u(k)) \| < \infty$. These conditions ensure that the system is non-singular at any time and for any initial condition.

The following definition of UOC adapts the uniform convergence (UC) property \cite[Def.~3]{tran2018convergence} to systems with output of the form \eqref{eq:DT-sys}.

\begin{definition} \label{defn:steady-state}
For any $u \in l_{n_u}^{\infty}$, a solution $x^*(k)$ to \eqref{eq:DT-sys} and its corresponding output $y^*(k) = h(k, x^*(k), u(k))$ are a reference state solution and the corresponding reference output, respectively, if they are defined for all $k \in \mathbb{Z}$, with $\| x^* \|_\infty < \infty$ and $\| y^* \|_{\infty} < \infty$.
\end{definition}

\begin{definition}
\label{defn:uoc}
System~\eqref{eq:DT-sys} is uniformly output convergent (UOC) if, for any input $u \in l^\infty_{n_u}$,
\begin{itemize}
\item[(i)] There exists a unique reference state solution $x^*$ with its corresponding reference output $y^*$.
\item[(ii)] There exists $\beta \in \mathcal{KL}$ independent of $u$ such that, for any $k, k_0 \in \mathbb{Z}$ with $k \geq k_0$ and $x(k_0) \in \mathbb{R}^{n_x}$,
\end{itemize}
\begin{equation}
\label{eq:defn_uoc}
\| y^*(k) - y(k) \| \leq \beta(\| x^*(k_0) - x(k_0) \|, k - k_0).
\end{equation}
System~\eqref{eq:DT-sys} with $y(k) = x(k)$ is uniformly convergent (UC) if, for any $u \in l^\infty_{n_u}$, 
\begin{itemize}
\item[(i)] There exists a unique reference state solution $x^*$.
\item[(ii)] There exists $\beta \in \mathcal{KL}$ independent of $u$ such that, for any $k, k_0 \in \mathbb{Z}$ with $k \geq k_0$ and $x(k_0) \in \mathbb{R}^{n_x}$,
\end{itemize}
\begin{equation}
\label{eq:defn_uc}
\| x^*(k) - x(k) \| \leq \beta(\| x^*(k_0) - x(k_0) \|, k - k_0).
\end{equation}
\end{definition}

\begin{remark} 
In this note, all gain functions (e.g. $\beta$ in \eqref{eq:defn_uoc} and \eqref{eq:defn_uc}) are independent of the input $u \in l_{n_u}^{\infty}$. Further, as in \cite{tran2018convergence}, `uniform' means that for each $x(k_0)$, the bound $\beta(\|x(k_0)\|, k-k_0)$ in \eqref{eq:defn_uoc} and \eqref{eq:defn_uc} depends on $k-k_0$ but not $k_0$.
\end{remark}

The UIOC property further extends the UOC property, and ensures that the perturbation in the reference outputs is asymptotically bounded by a nonlinear gain of the input perturbation.
The following definition of UIOC is a discrete-time analogue of the UISC property defined in \cite[Def.~3]{pavlov2005convergent} adapted to systems with output of the form \eqref{eq:DT-sys}.

\begin{definition}
\label{defn:uioc}
System \eqref{eq:DT-sys} is uniformly input-to-output convergent (UIOC) if it is UOC and for any $u, \overline{u} \in l^\infty_{n_u}$, with the reference state solution $x^*$ and its reference output $y^*$ associated to $u$, and any solution $\overline{x}(k)$ with any initial condition $\overline{x}(k_0)$ and the corresponding output $\overline{y}(k)$ associated to $\overline{u}$, there exists $\beta \in \mathcal{KL}, \gamma \in \mathcal{K}$ such that, for all $k_0, k \in \mathbb{Z}$ with $k \geq k_0$,
\begin{equation}
\label{eq:defn-uioc}
\begin{split}
\| y^*(k) - \overline{y}(k) \| \leq \max \{ \beta(\|x^*(k_0) - \overline{x}(k_0) \|, k - k_0), & \\
\gamma \left(\| (u - \overline{u})_{[k_0, k]} \| \right) \}. & 
\end{split}
\end{equation}
System~\eqref{eq:DT-sys} with $y(k) = x(k)$ is uniformly input-to-state convergent (UISC) if it is UC and
\begin{equation} \label{eq:defn-uisc}
\begin{split}
\| x^*(k) - \overline{x}(k) \| \leq \max \{ \beta(\|x^*(k_0) - \overline{x}(k_0) \|, k - k_0), & \\
 \gamma \left(\| (u - \overline{u})_{[k_0, k-1]} \| \right) \}.&
\end{split}
\end{equation}
\end{definition}
We conclude this section by summarizing the aforementioned stability concepts and their acronyms in Table~\ref{table1:acronym}.
\begin{table}[!ht]
\centering
\caption{Summary of stability concepts and their acronyms.}
\label{table1:acronym}
\begin{tabular}{ c | c  | c } 
Stability concept & Acronym & Definition \\
\hline\hline
Uniformly output convergent      & UOC  & Definition~\ref{defn:uoc} \\
\hline 
Uniformly convergent             & UC   & Definition~\ref{defn:uoc}  \\
\hline
Uniformly input-to-output convergent & UIOC & Definition~\ref{defn:uioc} \\
\hline
Uniformly input-to-state convergent  & UISC & Definition~\ref{defn:uioc} \\
\end{tabular}
\end{table}

\section{A UIOC small-gain theorem}
\label{sec:sg}
In this section, we present our main UIOC small-gain theorem (Theorem~\ref{theorem:uioc}). We first set some preliminaries.

For $k \in \mathbb{Z}$, consider the interconnected system (see Fig.~\ref{fig1}),
\begin{equation}
\label{eq:cl-sys}
\begin{split}
& \begin{cases}
x_1(k+1) = f_1(x_1(k), v_1(k), u_1(k)) \\
\hspace{1.8em} y_1(k) = h_1(x_1(k), v_1(k), u_1(k)), \\
\end{cases} \\
&\begin{cases}
x_2(k+1) = f_2(x_2(k), v_2(k), u_2(k)) \\
\hspace{1.8em} y_2(k) = h_2(x_2(k), v_2(k), u_2(k)), \\
\end{cases} \\
& \qquad v_1(k) = y_2(k), \ v_2(k) = y_1(k).
\end{split}
\end{equation}

\begin{figure}[!ht]
\centering
\includegraphics[scale=0.9, trim={0em, 0em, 0em, 0em}, clip]{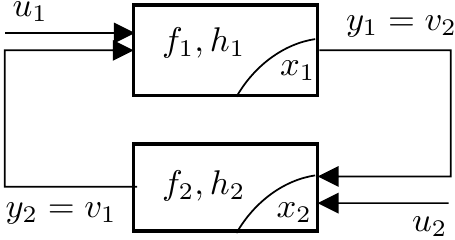}
\caption{Schematic of the interconnected system~\eqref{eq:cl-sys}.}
\label{fig1}
\end{figure}

For subsystems $j=1,2$, $x_j(k) \in \mathbb{R}^{n_{x_j}}$ are states, $y_j(k) \in \mathbb{R}^{n_{y_j}}$ are outputs and $u_j(k) \in \mathbb{R}^{n_{u_j}}$ are inputs with $u_j \in l^{\infty}_{n_{u_j}}$. Throughout, we assume that the interconnected system~\eqref{eq:cl-sys} is well-posed. That is, for any $k, k_0 \in \mathbb{Z}$ with $k \geq k_0$ and any initial conditions $x_j(k_0)$, there exists a unique solution $y(k) \coloneqq (y_1(k), y_2(k)) \in \mathbb{R}^{n_{y_1} + n_{y_2}}$ solving the algebraic equations $ y_1(k)  = h_1(x_1(k), y_2(k), u_1(k))$ and $y_2(k)  = h_2(x_2(k), y_1(k), u_2(k))$. 

For a well-posed system~\eqref{eq:cl-sys}, let $x(k) \coloneqq (x_1(k), x_2(k)) \in \mathbb{R}^{n_{x_1} + n_{x_2}}$ be the closed-loop solution to input $u(k) \coloneqq (u_1(k),$  $ u_2(k)) \in \mathbb{R}^{n_{u_1} + n_{u_2}}$, starting at $x(k_0) = (x_1(k_0), x_2(k_0))$. Note that the closed-loop system is causal by definition. 

We also assume that the subsystems in \eqref{eq:cl-sys} are UOC. Each UOC subsystem induces an I/O map $\mathcal{F}_j: l^{\infty}_{n_{v_j}} \times l^{\infty}_{n_{u_j}} \rightarrow l^{\infty}_{n_{y_j}}$ defined by $\mathcal{F}_j(v_j, u_j) = y^*_j$, where $y^*_j$ is the reference output. By construction $\mathcal{F}_j$ is causal, meaning that for any $\tau \in \mathbb{Z}$ and any $v_j \in l^{\infty}_{n_{v_j}}, u_j \in l^{\infty}_{n_{u_j}}$,
\begin{equation} \label{eq:causality}
\Pi_{\tau} \circ \mathcal{F}_j \circ \Pi_{\tau}(v_j, u_j) = \Pi_{\tau} \circ \mathcal{F}_j(v_j, u_j),
\end{equation}
where $\Pi_{\tau}(v_j, u_j) = (v_j(k), u_j(k))$ for $k \leq \tau$ and zero otherwise. We say that system~\eqref{eq:cl-sys} induces a well-posed closed-loop I/O map if the algebraic equations $y_1 = \mathcal{F}_2(y_2, u_2)$ and $y_2 = \mathcal{F}_1(y_1, u_1)$ have a unique bounded solution $y^*_{cl} = (y^*_{1, cl}, y^*_{2, cl}) \in l^\infty_{n_{y_1} + n_{y_2}}$, and the closed-loop I/O map $\mathcal{F}(u) =  y^*_{cl}$ is causal \cite{shamma1993fading}, where $u=(u_1, u_2) \in l^{\infty}_{n_{u_1} + n_{u_2}}$. We emphasize the difference between a well-pose  system \eqref{eq:cl-sys} and a well-posed closed-loop I/O map induced by \eqref{eq:cl-sys}.

We now state our main UIOC small-gain theorem.

\begin{theorem}
\label{theorem:uioc}
Consider a well-posed system~\eqref{eq:cl-sys} with UOC subsystems $j=1,2$. For any inputs $u_j \in l^\infty_{n_{u_j}}, v_j \in l^\infty_{n_{v_j}}$, let $x_j^*$ and $y^*_j$ be the corresponding reference state solutions and outputs. For any other inputs $\overline{u}_j, \overline{v}_j$ with $\overline{u}_j \in l^{\infty}_{n_{u_j}}$, let $\overline{x}_j$ and $\overline{y}_j$ be any corresponding solutions and outputs with initial conditions $\overline{x}_j(k_0)$. Suppose that there exists $\beta_j \in \mathcal{KL}$ and $\gamma^y_j, \gamma^u_j, \sigma_j, \sigma^u_j, \sigma^y_j \in \mathcal{K}$ (independent of $u_j, v_j$) such that, for all $k_0 , k \in \mathbb{Z}$, $k \geq k_0$ and any $\overline{x}_j(k_0) \in \mathbb{R}^{n_{x_j}}$,
\begin{equation}
\label{eq:theorem-uioc-y}
\begin{split}
&\| y^*_j(k) -  \overline{y}_j(k) \|\\
& \leq \max \{ \beta_j(\| x^*_j(k_0) - \overline{x}_j(k_0)\|, k - k_0),\\
& \hspace*{3.5em} \gamma_j^y( \|(v_j - \overline{v}_j)_{[k_0, k]}\|), \gamma^u_j(\| (u_{j} - \overline{u}_j)_{[k_0, k]}\|)\}, 
\end{split}
\end{equation}
\begin{equation}
\label{eq:theorem-uioc-x}
\begin{split}
& \| x^*_j(k) -  \overline{x}_j(k) \| \leq \max \{\sigma_j(\| x^*_j(k_0) - \overline{x}_j(k_0)\|), \\
& \hspace*{1em} \sigma^y_j(\| (v_j -  \overline{v}_{j})_{[k_0, k-1]}\|), \sigma^u_j(\| (u_{j} - \overline{u}_{j})_{[k_0, k-1]} \|) \}.
\end{split}
\end{equation}
If $\gamma_1^y \circ \gamma_2^y (s) < s$ (or equivalently $\gamma_2^y \circ \gamma_1^y (s) < s$ \cite[Chapter~8.1]{isidori2017lectures}) for all $s > 0$, for any $k_0 \in \mathbb{Z}$ and $x(k_0) \in \mathbb{R}^{n_{x_1} + n_{x_2}}$, the closed-loop solution and its output are bounded, i.e., $\sup_{k \geq k_0} \| x(k)\| < \infty$ and $\sup_{k \geq k_0} \| y(k) \| < \infty$. Further, the closed-loop system~\eqref{eq:cl-sys} induces a well-posed closed-loop I/O map and system~\eqref{eq:cl-sys} is UIOC.
\end{theorem}

\begin{remark}
Although Theorem~\ref{theorem:uioc} ensures the closed-loop solution $x(k)$ is bounded, in general, it does not guarantee the closed-loop system~\eqref{eq:cl-sys} is UISC, which would require additional conditions; see Corollary~\ref{corollary:uisc}.
\end{remark}

The main idea in the proof of Theorem~\ref{theorem:uioc} is to use the Banach fixed point theorem to show system~\eqref{eq:cl-sys} induces a well-posed closed-loop I/O map. To show that system~\eqref{eq:cl-sys} is UIOC, we apply a change-of-coordinate argument in which the system under consideration becomes time-varying of the form \eqref{eq:uios-cl-sys}. We then apply a uniform input-to-output stability (UIOS) small-gain theorem (Theorem~\ref{theorem:uios}) to system~\eqref{eq:uios-cl-sys} to establish the UIOC property of \eqref{eq:cl-sys}. We first define UIOS and state Theorem~\ref{theorem:uios} whose full proof is given in Appendix~\ref{app:theorem_uios}.

\begin{definition}
\label{defn:UIOS}
System~\eqref{eq:DT-sys} is uniformly input-to-output stable (UIOS) if there exists $\beta \in \mathcal{KL}$ and $\gamma \in \mathcal{K}$ such that, for any $u \in l^\infty_{n_u}$, $k, k_0 \in \mathbb{Z}$ with $k \geq k_0$ and $x(k_0) \in \mathbb{R}^{n_x}$,
\begin{equation} \label{eq:defn_UIOS} 
\|y(k) \| \leq \max \left\{ \beta \left(\|x(k_0)\|, k-k_0 \right) , \gamma\left( \left\| u_{[k_0, k]} \right\| \right)  \right\}.
\end{equation}
\end{definition}

\begin{theorem}
\label{theorem:uios}
Consider a well-posed time-varying system
\begin{equation}
\label{eq:uios-cl-sys}
\begin{split}
& \begin{cases}
\Delta x_1(k+1) = \tilde{f}_1(k, \Delta x_1(k), \Delta v_1(k), \Delta u_1(k)) \\
\hspace{1.8em} \Delta y_1(k) = \tilde{h}_1(k, \Delta x_1(k), \Delta v_1(k), \Delta u_1(k)), \\
\end{cases} \\
&\begin{cases}
\Delta x_2(k+1) = \tilde{f}_2(k, \Delta x_2(k), \Delta v_2(k), \Delta u_2(k)) \\
\hspace{1.8em} \Delta y_2(k) = \tilde{h}_2(k, \Delta x_2(k), \Delta v_2(k),  \Delta u_2(k)), \\
\end{cases} \\
& \qquad \Delta v_1(k) = \Delta y_2(k), \ \Delta v_2(k) = \Delta y_1(k).
\end{split}
\end{equation}
For $j=1,2$, suppose that there exists $\beta_j \in \mathcal{KL}$ and $\gamma^y_j, \gamma^u_j, \sigma_j, \sigma^u_j, \sigma^y_j \in \mathcal{K}$ such that, for any $\Delta v_j, \Delta u_j$ with $\Delta u_j \in l^{\infty}_{ n_{\Delta u_j}}$, $k_0, k \in \mathbb{Z}$ with $k \geq k_0$ and $\Delta  x_j(k_0) \in \mathbb{R}^{n_{x_j}}$,
\begin{equation}
\label{eq:theorem-uios-y}
\begin{split}
\| \Delta y_j(k) \| \leq \max \{ & \beta_j(\|\Delta x_j(k_0) \|, k-k_0), \\
& \gamma^y_j(\|\Delta v_{j_{[k_0, k]}} \|), \gamma^u_j(\| \Delta u_{j_{[k_0, k]}} \|) \} ,
\end{split}
\end{equation}
\begin{equation}
\label{eq:theorem-uios-x}
\begin{split}
\| \Delta x_j (k) \| \leq & \max \{  \sigma_j (\| \Delta x_j (k_0) \|), \\
& \sigma^y_j(\| \Delta v_{j_{[k_0, k-1]}} \|), \sigma^u_j(\|\Delta u_{j_{[k_0, k-1]}} \|) \}. 
\end{split}
\end{equation}
If $\gamma^y_1 \circ \gamma^y_2 (s) < s$ (or equivalently $\gamma^y_2 \circ \gamma^y_1 (s) < s$) for all $s>0$, then for any $k_0 \in \mathbb{Z}$ and $\Delta x(k_0) \in \mathbb{R}^{n_{x_1} + n_{x_2}}$, the closed-loop solution and its output are bounded, i.e., $\sup_{k \geq k_0} \| \Delta x(k) \| < \infty$ and $\sup_{k \geq k_0} \| \Delta y(k) \| < \infty$. Furthermore, the closed-loop system~\eqref{eq:uios-cl-sys} is UIOS.
\end{theorem}

We now detail the proof for Theorem~\ref{theorem:uioc}.

\begin{proof}[Proof of Theorem~\ref{theorem:uioc}]
For any fixed inputs $u_j \in l^\infty_{n_{u_j}}$, the I/O map induced by each subsystem is given by $\mathcal{F}^{u_j}_j: l_{n_{v_j}}^\infty \rightarrow l^\infty_{n_{y_j}}, \mathcal{F}^{u_j}_j(v_j) = y^*_j$. For any $v_j, \overline{v}_j \in l^\infty_{n_{v_j}}$, let $k_0 \rightarrow -\infty$ and take the supremum over $k \in \mathbb{Z}$ in \eqref{eq:theorem-uioc-y},
\begin{equation} \label{eq:uioc-proof1}
\| \mathcal{F}^{u_j}_j(v_j) - \mathcal{F}^{u_j}_j(\overline{v}_j) \|_{\infty} \leq \gamma^y_j(\|v_j -\overline{v}_j\|_{\infty}).
\end{equation}
Consider the composition $\mathcal{F}^{u_1}_1 \circ \mathcal{F}^{u_2}_2 : l^{\infty}_{n_{v_2}} \rightarrow l^{\infty}_{n_{v_2}}$. Applying inequality \eqref{eq:uioc-proof1} twice, we have
\begin{equation*}
\begin{split}
& \| \mathcal{F}^{u_1}_1 \circ \mathcal{F}^{u_2}_2 (v_2) - \mathcal{F}^{u_1}_1 \circ \mathcal{F}^{u_2}_2(\overline{v}_2) \|_{\infty} \\
& \quad \leq \gamma^y_1 \circ \gamma^y_2 (\| v_2 - \overline{v}_2\|_{\infty}) < \| v_2 - \overline{v}_2 \|_{\infty}.
\end{split}
\end{equation*}
Therefore, $\mathcal{F}^{u_1}_1 \circ \mathcal{F}^{u_2}_2$ is a strict contraction on  $(l^\infty_{n_{v_2}}, \| \cdot \|_\infty)$. Its unique fixed point $y^*_{1, cl} \in l^{\infty}_{n_{v_2}}$ given by the Banach fixed-point theorem \cite{kreyszig1978introductory} is the reference output of subsystem $j=1$. The corresponding reference output of subsystem $j=2$ is $y^*_{2, cl} = \mathcal{F}^{u_2}_2(y^*_{1, cl})$. A symmetric argument shows that $\mathcal{F}^{u_2}_2 \circ \mathcal{F}^{u_1}_1$ is a strict contraction defined on $(l^{\infty}_{n_{v_1}}, \| \cdot \|_{\infty})$.

To show that the closed-loop I/O map is causal, for any $\tau \in \mathbb{Z}$, consider $y^*_{1, \Pi_{\tau}} \in l^{\infty}_{n_{v_2}}$ the unique fixed point of $\mathcal{F}^{\Pi_{\tau}(u_1)}_1 \circ \mathcal{F}^{\Pi_{\tau}(u_2)}_2$. Causality follows if $\Pi_{\tau}(y^*_{1, \Pi_\tau}) = \Pi_{\tau}(y^*_{1, cl})$. Re-express \eqref{eq:causality} as $\Pi_{\tau} \circ \mathcal{F}^{\Pi_{\tau}(u_j)}_j \circ \Pi_{\tau} (v_j) = \Pi_{\tau} \circ \mathcal{F}^{u_j}_j(v_j)$,
\begin{equation*}
\begin{split}
\Pi_\tau (y^*_{1, cl})& =  \Pi_\tau \circ \mathcal{F}^{u_1}_1 \circ \mathcal{F}^{u_2}_2(y^*_{1, cl}) \\
& = \Pi_\tau \circ \mathcal{F}^{\Pi_\tau(u_1)}_1 \circ \Pi_\tau (\mathcal{F}^{u_2}_2(y^*_{1, cl}))  \\
& = \Pi_\tau \circ \mathcal{F}^{\Pi_\tau(u_1)}_1 \circ \Pi_\tau \circ \mathcal{F}^{\Pi_\tau(u_2)}_2 \circ \Pi_\tau (y^*_{1, cl}) \\
& = \Pi_\tau \circ \mathcal{F}^{\Pi_\tau(u_1)}_1 \circ \mathcal{F}^{\Pi_\tau(u_2)}_2 \circ \Pi_\tau (y^*_{1, cl}).
\end{split}
\end{equation*}
Similarly, $\Pi_\tau(y^*_{1, \Pi_\tau})  =  \Pi_{\tau} \circ \mathcal{F}^{\Pi_\tau(u_1)}_1 \circ  \mathcal{F}^{\Pi_\tau(u_2)}_2 \circ \Pi_\tau (y^*_{1, \Pi_\tau})$. Note that $\Pi_{\tau} \circ \mathcal{F}^{\Pi_\tau(u_1)}_1 \circ \mathcal{F}^{\Pi_\tau(u_2)}_2$ is a strict contraction, by the Banach fixed point theorem we have $\Pi_\tau(y^*_{1, cl}) = \Pi_\tau(y^*_{1, \Pi_\tau})$.

To establish closed-loop UIOC, let $x^*_{1, cl}$ be the reference state solution to subsystem $j=1$ with respect to input $(y^*_{2, cl}, u_1)$ and analogously for $x^*_{2,cl}$. Then $x^*_{cl} = (x^*_{1, cl}, x^*_{2, cl})$ is the closed-loop reference state solution.
Let $\overline{x}(k) = (\overline{x}_1(k), \overline{x}_2(k)), \overline{y}(k) = (\overline{y}_1(k), \overline{y}_2(k))$ be any other closed-loop solution and its corresponding output to another input $\overline{u} = (\overline{u}_1, \overline{u}_2)$, starting at $\overline{x}(k_0)$. From \eqref{eq:theorem-uioc-y} and \eqref{eq:theorem-uioc-x}, the same argument as in the proof of Theorem~\ref{theorem:uios} shows that $\sup_{k \geq k_0} \|x_{cl}^*(k) - \overline{x}(k)\| < \infty$, $\sup_{k \geq k_0}  \|y^*_{cl}(k) - \overline{y}(k) \| < \infty$.

For $j=1, 2$ and $k \geq k_0$, define $\Delta x_j(k) = \overline{x}_j(k) - x^*_{j, cl}(k)$, $\Delta y_j(k) = \overline{y}_j(k) - y^*_{j, cl}(k)$ and $\Delta u_j(k) = \overline{u}_j(k)- u_j(k)$. Let $v^*_{1, cl} = y^*_{2, cl}$ and $v^*_{2, cl} = y^*_{1, cl}$, we have the time-varying systems with states $\Delta x_j(k)$, outputs $\Delta y_j(k)$, inputs $\Delta u_j(k)$, and interconnections $\Delta v_1(k) = \Delta y_2(k)$ and $\Delta v_2(k) = \Delta y_1(k)$,
\begin{equation*}
\begin{split}
& \Delta x_j(k+1) \\
& = f_j(\Delta x_j(k) + x^*_{j, cl}(k), \Delta v_j(k) + v^*_{j, cl}(k), \Delta u_j(k) + u_j(k)) \\
& \qquad - f_j(x^*_{j, cl}(k), v^*_{j, cl}(k), u_j(k)) \\
& = \tilde{f}_j(k, \Delta x_j(k), \Delta v_j(k), \Delta u_j(k)), \\
\end{split}
\end{equation*}
\begin{equation*}
\begin{split}
& \Delta y_j(k) \\
& = h_j(\Delta x_j(k) + x^*_{j, cl}(k), \Delta v_j(k) + v^*_{j, cl}(k), \Delta u_j(k) + u_j(k)) \\
& \qquad - h_j(x^*_{j, cl}(k), v^*_{j, cl}(k), u_j(k)) \\
& = \tilde{h}_j(k, \Delta x_j(k), \Delta v_j(k), \Delta u_j(k)).
\end{split}
\end{equation*}
From \eqref{eq:theorem-uioc-y} and \eqref{eq:theorem-uioc-x}, $\Delta x_j(k)$ and $\Delta y_j(k)$ satisfy \eqref{eq:theorem-uios-y} and \eqref{eq:theorem-uios-x} in Theorem~\ref{theorem:uios}. Finally, closed-loop UIOC follows from applying Theorem~\ref{theorem:uios} to the above interconnected systems.
\end{proof}

Sometimes it is convenient to upper bound $\| y^*_j(k) - \overline{y}(k)\|$ in \eqref{eq:theorem-uioc-y} by a sum instead of $\max$ of nonlinear gains. That is, 
\begin{equation} \label{eq:alternative-uioc}
\begin{split}
\| y^*_j(k) - \overline{y}_j(k) \| \leq  \beta_j(\|x^*_j(k_0) - \overline{x}_j(k_0) \|, k-k_0) &\\
 + \gamma^y_j(\| (v_j - \overline{v}_j)_{[k_0, k]} \|) + \gamma^u_j(\| (u_j - \overline{u}_j)_{[k_0, k]} \|).&
\end{split}
\end{equation}
Theorem~\ref{theorem:uioc} can be applied to this scenario by re-writing \eqref{eq:alternative-uioc} in terms of $\max$ of nonlinear gains. If system~\eqref{eq:cl-sys} satisfies \eqref{eq:alternative-uioc} instead of \eqref{eq:theorem-uioc-y}, to ensure UIOC of \eqref{eq:cl-sys}, the condition $\gamma^y_1 \circ \gamma^y_2(s)<s$ for all $s > 0$ needs to be strengthened to \eqref{eq:coro-uios-plus} below. We first present a lemma that allows us to re-write \eqref{eq:alternative-uioc} in terms of $\max$.

\begin{lemma} \label{lemma:plus-extension}
Given any $\lambda \in \mathcal{K}_{\infty}$, for any $a, b \geq 0$, it holds that $a + b \leq \max \{ a + \lambda(a), b + \lambda^{-1}(b)\}$.
\end{lemma}
\begin{proof}
Since $\lambda \in \mathcal{K}_{\infty}$, its inverse $\lambda^{-1}$ exists and is in $\mathcal{K}_{\infty}$. Consider two cases. Suppose $b \leq \lambda(a)$, then $a+b \leq a + \lambda(a) \leq \max \{ a + \lambda(a), b + \lambda^{-1}(b)\}$. Otherwise, $b > \lambda(a)$. Applying $\lambda^{-1}$ on both sides gives $\lambda^{-1}(b) > a$ and $a+b < \lambda^{-1}(b) + b \leq \max\{ a+\lambda(a), b + \lambda^{-1}(b) \}$.
\end{proof}

Theorem~\ref{theorem:uioc} and Lemma~\ref{lemma:plus-extension} leads to the following Corollary.
\begin{corollary}
\label{corollary:uioc-plus}
Consider a well-posed system~\eqref{eq:cl-sys} with UOC subsystems. For any inputs $u_j \in l^\infty_{n_{u_j}}, v_j \in l^\infty_{n_{v_j}}$, let $x_j^*$ and $y^*_j$ be the corresponding reference state solutions and outputs. For any other inputs $\overline{u}_j, \overline{v}_j$ with $\overline{u}_j \in l^{\infty}_{n_{u_j}}$, let $\overline{x}_j$ and $\overline{y}_j$ be any corresponding solutions and outputs with initial conditions $\overline{x}_j(k_0)$. Suppose that there exists $\beta_j \in \mathcal{KL}$ and $\gamma^y_j, \gamma^u_j, \sigma_j, \sigma^u_j, \sigma^y_j \in \mathcal{K}$ such that, for all $k_0 , k \in \mathbb{Z}$ with $k \geq k_0$ and any $\overline{x}_j(k_0) \in \mathbb{R}^{n_{x_j}}$, \eqref{eq:theorem-uioc-x} and \eqref{eq:alternative-uioc} hold. If there exists $\lambda_j \in \mathcal{K}_{\infty}$ such that for all $s >0$,
\begin{equation} \label{eq:coro-uios-plus}
\begin{split}
&(id + \lambda_1) \circ \gamma^y_1 \circ (id + \lambda_2) \circ \gamma^y_2 (s) < s, \\
\end{split}
\end{equation}
where $id$ is the identity map. Then for any $k_0 \in \mathbb{Z}$ and $x(k_0) \in \mathbb{R}^{n_{x_1} + n_{x_2}}$, the closed-loop solution and output are bounded, i.e., $\sup_{k \geq k_0}\|x(k)\| < \infty$ and $\sup_{k \geq k_0} \| y(k) \| < \infty$. Further, system~\eqref{eq:cl-sys} is UIOC and induces a well-posed closed-loop I/O map.
\end{corollary}

We can further apply Theorem~\ref{theorem:uioc} and Lemma~\ref{lemma:plus-extension} to system~\eqref{eq:cl-sys} with $y_j(k) = x_j(k)$. In this case, Theorem~\ref{theorem:uioc} ensures the UISC property of system~\eqref{eq:cl-sys}, leading to the following Corollary.
\begin{corollary} \label{corollary:uisc}
Consider a well-posed system~\eqref{eq:cl-sys} with $y_j(k) = x_j(k)$ and UC subsystems. For any inputs $u_j \in l^\infty_{n_{u_j}}, v_j \in l^\infty_{n_{v_j}}$, let $x_j^*$ be the corresponding reference state solutions. For any other inputs $\overline{u}_j, \overline{v}_j$ with $\overline{u}_j \in l^{\infty}_{n_{u_j}}$, let $\overline{x}_j$ be any corresponding solutions with initial conditions $\overline{x}_j(k_0)$. Suppose that there exists $\beta_j \in \mathcal{KL}$ and $\gamma^y_j, \gamma^u_j \in \mathcal{K}$ such that, for all $k_0 , k \in \mathbb{Z}$, $k \geq k_0$ and any $\overline{x}_j(k_0) \in \mathbb{R}^{n_{x_j}}$,
\begin{equation} \label{eq:coro-uisc-plus}
\begin{split}
& \| x^*_j(k) -  \overline{x}_j(k) \| \leq \beta_j( \| x^*_j(k_0) - \overline{x}_j(k_0)\|, k - k_0)  \\
& + \gamma_j^y( \|(v_j - \overline{v}_j)_{[k_0, k-1]}\|) + \gamma^u_j(\| (u_{j} - \overline{u}_j)_{[k_0, k-1]}\|).
\end{split}
\end{equation}
If \eqref{eq:coro-uios-plus} holds, then for any $k_0 \in \mathbb{Z}$ and $x(k_0) \in \mathbb{R}^{n_{x_1} + n_{x_2}}$, the closed-loop solution is bounded, i.e., $\sup_{k \geq k_0}\|x(k)\| < \infty$. Further, system~\eqref{eq:cl-sys} is UISC.
\end{corollary}

\section{Applications}
\label{sec:applications}
This section demonstrates potential applications of Theorem~\ref{theorem:uioc} on observer design and black-box system identification. 

\subsection{Observer-based controller design} 
\label{subsec:obs}
The conventional observer-based controller design approach first finds a desired solution of the closed-loop system $\tilde{x}(k)$. It then ensures that for any other solution $x(k)$, $z(k) = x(k) - \tilde{x}(k)$ is asymptotically stable. Since $z(k)$ typically depends on $\tilde{x}(k)$, it can be difficult to analyze the asymptotic stability of $z(k)$. Recently, the convergence approach has been applied to observer-based control in continuous-time, which may circumvent the cumbersome stability analysis in the conventional approach \cite{pavlov2005convergent}. The convergence approach first designs an observer and a feedback controller so that the closed-loop system is UC. This ensures that the closed-loop system induces a well-posed I/O map and the internal states are bounded. Secondly, a feedforward controller is designed to shape the closed-loop system's response. Here, we employ Theorem~\ref{theorem:uioc} to achieve the first step in the convergence approach.

For $k \in \mathbb{Z}$, consider a nonlinear plant
\begin{equation*} 
z(k+1) = f(z(k), u(k), w(k)), \quad y(k) = h(z(k)),
\end{equation*}
with state $z(k) \in \mathbb{R}^{n_z}$, control $u(k) \in \mathbb{R}^{n_u}$, external input $w \in l^\infty_{n_w}$ and output $y(k) \in \mathbb{R}^{n_y}$. Construct an observer,
\begin{equation*} 
\hat{z}(k+1) = f(\hat{z}(k), u(k), w(k)) - L(h(\hat{z}(k)) - y(k)),
\end{equation*}
where $u(k) = \phi(\hat{z}(k))$ and $L(s)=0$ if $s = 0$. In general, $L$ can be a nonlinear function. Let $\Delta z(k) = \hat{z}(k) - z(k)$, then $u(k) = \phi(\Delta z(k) + z(k))$ and the observer error dynamics is
\begin{equation*}
\begin{split}
\Delta z(k+1)  & = \hat{f}(\Delta z(k), z(k), w(k)) \\
& = f(\Delta z(k) + z(k), u(k), w(k)) \\ 
& \hspace*{-3em} - f(z(k), u(k), w(k))  - L(h(\Delta z(k) + z(k)) - h(z(k))), \\
\end{split}
\end{equation*}
where $z(k), w(k)$ are viewed as inputs to the error dynamics. Consider the interconnected system~\eqref{eq:app-cl} (see Fig.~\ref{fig2}),
\begin{equation} \label{eq:app-cl}
\begin{split}
& \begin{cases}
\hspace*{0.8em} z(k+1)  = f(z(k), \phi(v_1(k) + z(k)), w(k)), \\
\Delta z(k+1) = \hat{f}(\Delta z(k), v_2(k), w(k)),
\end{cases} \\
& \hspace*{3.15em} v_1(k) = \Delta z(k), \quad  v_2(k) = z(k),
\end{split}
\end{equation}
where $w$ is the input. Our goal is to employ Theorem~\ref{theorem:uioc} to design the observer gain $L(\cdot)$ and the controller $u(k) = \phi(\hat{z}(k))$ such that the closed-loop system is UISC. To achieve this goal, we employ the following Corollary of Theorem~\ref{theorem:uioc} to system~\eqref{eq:app-cl}, whose full proof is given in Appendix~\ref{app:obs_coro}.

\begin{figure}[!ht]
\centering
\includegraphics[scale=0.9]{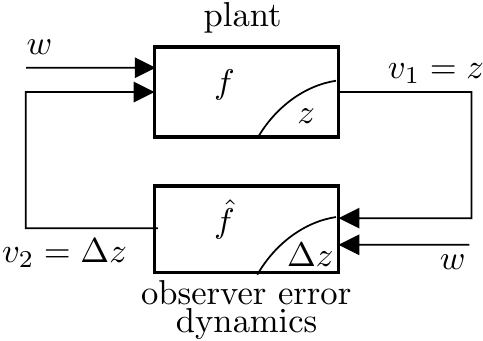}
\caption{Schematic of the closed-loop system~\eqref{eq:app-cl} consisting of the plant and the observer error dynamics.}
\label{fig2}
\end{figure}

\begin{corollary} \label{corollary:obs}
Consider a well-posed system~\eqref{eq:app-cl}. Suppose that for any inputs $v_2, w$ with $w \in l^{\infty}_{n_w}$, there exists $\beta_2 \in \mathcal{KL}$ such that, for any $k, k_0 \in \mathbb{Z}, k \geq k_0$ and $\Delta z(k_0) \in \mathbb{R}^{n_{\Delta z}}$,
\begin{equation} \label{eq:coro-uioc-obs}
\| \Delta z(k) \| \leq \beta_2( \| \Delta z(k_0) \|, k-k_0).
\end{equation}
Suppose that the $z$-subsystem is UC and let $z^*$ be the reference solution to $v_1, w$. For any other input $\overline{v}_1, \overline{w}$ with $\overline{w} \in l^{\infty}_{n_w}$, let $\overline{z}$ be any solution. Suppose that there exists $\beta_1 \in \mathcal{KL}, \gamma_1^y, \gamma_1^w \in \mathcal{K}$ such that for any $k_0, k \in \mathbb{Z}, k \geq k_0$ and $\overline{z}(k_0) \in \mathbb{R}^{n_z}$,
\begin{equation} \label{eq:coro-uioc-obs-x1}
\begin{split}
& \| z^*(k) - \overline{z}(k) \| \leq  \beta_1(\|z^*(k_0) - \overline{z}(k_0)\|, k-k_0) \\
& +   \gamma^y_1(\| (v_1 - \overline{v}_1)_{[k_0, k-1]} \|) + \gamma^w_1(\| (w - \overline{w})_{[k_0, k-1]} \|).
\end{split}
\end{equation}
It follow that system~\eqref{eq:app-cl} is UISC.
\end{corollary}

As a concrete example, we employ Corollary~\ref{corollary:obs} to a design observer-based controller for a Lur'e system with a globally Lipschitz nonlinearity modified from \cite[Example~1]{ibrir2007circle},
\begin{equation*}
z(k+1) = A z(k) + B_u u(k) + B_w w(k) + \rho G \sin(H z(k)),
\end{equation*}
with output $y(k)  = C z(k)$, $\rho = 0.1$ and 
\begin{equation*}
\begin{split}
& A = \begin{bmatrix} 1 & 1 \\ 0 & 1.1\end{bmatrix}, B_u = \begin{bmatrix} 1 \\ 1\end{bmatrix}, B_w = \begin{bmatrix} -0.5 \\ 1 \end{bmatrix}, \\
& C = \begin{bmatrix}0.1 & 0.5 \end{bmatrix}, G = \begin{bmatrix} 0.5 \\ 1\end{bmatrix}, H = \begin{bmatrix} 1 & 1\end{bmatrix}.
\end{split}
\end{equation*}
Here, for any $a, \hat{a} \in \mathbb{R}^2$, we have
\begin{equation*}
\| \rho G \sin(H a) - \rho G \sin(H \hat{a}) \| \leq \|\rho GH (a - \hat{a}) \|.
\end{equation*}
Consider a Luenberger observer with gain $L = P^{-1} Z \in \mathbb{R}^{2}$,
\begin{equation*}
\begin{split}
\hat{z}(k+1) = & A \hat{z}(k) + B_u u(k) + B_w w(k) \\
& + \rho G \sin(H\hat{z}(k)) - P^{-1} Z C(\hat{z}(k) - z(k)),
\end{split}
\end{equation*}
where $P > 0$. In Appendix~\ref{app:obs_lmi}, we show that the observer error dynamics $\Delta z(k)  = \hat{z}(k) - z(k)$ satisfies \eqref{eq:coro-uioc-obs} if there exists $Z \in \mathbb{R}^2$, $P>0$, $\epsilon>0$ and $\theta \in (0, 1)$ such that
\begin{small}
\begin{equation} \label{eq:lmi}
\begin{split}
P - \epsilon I & < 0, \\
\begin{bmatrix}
-\theta P & A^\top P - C^\top Z^\top & \epsilon \rho (GH)^\top & A^\top P - C^\top Z^\top \\
PA - ZC & -P & 0 & 0 \\
\epsilon \rho GH & 0 & -\epsilon I & 0 \\
PA - ZC & 0 & 0 & P - \epsilon I
\end{bmatrix} & \leq 0.
\end{split}
\end{equation}
\end{small}

Consider a linear state-feedback law $u(k) = - K \hat{z}(k)$ with gain $K^\top \in \mathbb{R}^{2}$. The plant subject to $u(k)$ becomes
\begin{equation} \label{eq:app-lip-fb}
\begin{split}
z(k+1) & = (A - B_u K) z(k) + \rho G \sin(Hz(k)) \\
& \hspace*{1em} - B_u K \Delta z(k) + B_w w(k) \\
& \coloneqq \tilde{f}(z(k), \Delta z(k), w(k)),
\end{split}
\end{equation}
where $\Delta z(k), w(k)$ are viewed as inputs. Let $z(k), \overline{z}(k)$be any solutions starting at $z(k_0), \overline{z}(k_0)$ to inputs $\Delta z(k), w$ and $\Delta \overline{z}(k), \overline{w}$, respectively. Let $\delta \overline{z}(k) = z(k) - \overline{z}(k)$, then 
\begin{equation} \label{eq:obs1}
\begin{split}
\| \delta \overline{z}(k) \|  & \leq  \lambda_s \| \delta \overline{z}(k-1) \| \\
& + \sigma_{\max}(B_u K) \|\Delta z(k-1) - \Delta \overline{z}(k-1)\| \\
& + \|B_w\|  |w(k-1) - \overline{w}(k-1)|,
\end{split}
\end{equation}
where $\lambda_s = \sigma_{\max}(A - B_u K) + \rho \sigma_{\max} (GH) $. We employ \cite[Theorem~1]{pavlov2008convergent} to show that the plant~\eqref{eq:obs1} is UC. Firstly, consider $\Delta z = \Delta \overline{z}, w = \overline{w}$. From \eqref{eq:obs1}, we have $\| \delta \overline{z}(k) \| \leq \lambda_s \| \delta \overline{z}(k-1)\|$. Further, note that for any $\Delta z \in l^{\infty}_{2}$ and $w \in l^{\infty}_1$, we have
\begin{equation*}
\begin{split}
& \sup_{k \in \mathbb{Z}}\|\tilde{f}(0, \Delta z(k), w(k))\| \\
& \leq \| B_w \| \|w\|_{\infty}  + \sigma_{\max}(B_u K) \| \Delta z \|_{\infty} < \infty.
\end{split}
\end{equation*}
By \cite[Theorem~1]{pavlov2008convergent}, if there exists $K$ such that $\lambda_s < 1$, then the plant~\eqref{eq:app-lip-fb} is UC.

Furthermore, the condition $\lambda_s < 1$ also ensures that the plant satisfies \eqref{eq:coro-uioc-obs-x1} in Corollary~\ref{corollary:obs}. Finally, applying Corollary~\ref{corollary:obs} shows that the closed-loop system~\eqref{eq:app-cl} is UISC.

\begin{example} \label{example:lure-obs}
Choose $L^{\top}$ $=\begin{bmatrix} 2.3258 & 2.1104 \end{bmatrix}$ and $K = \begin{bmatrix} 0.4956 & 1.006 \end{bmatrix}$, then $\lambda_s = 0.8687$ and the linear matrix inequalities \eqref{eq:lmi} hold for $\theta = \epsilon = 0.001$. Hence, the closed-loop system~\eqref{eq:app-cl} is UISC. The UISC property ensures that all solutions $z(k)$ of the controlled plant to an input $w$ asymptotically converge to the reference state solution $z^*$, independent of initial condition $z(k_0)$; see Fig.~\ref{fig3}.

\begin{figure}[!ht]
\centering
\includegraphics[scale=1, trim={0.em 0.em 0em 0em},clip]{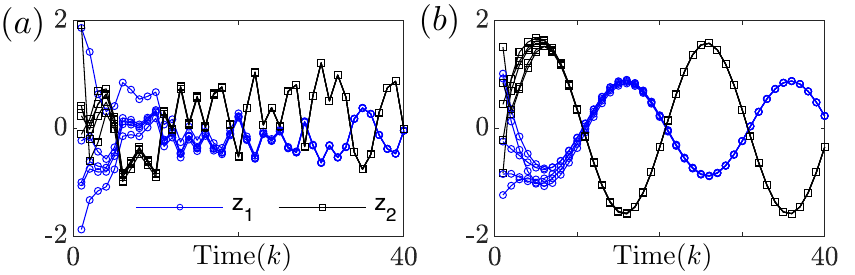}
\caption{The controlled plant states $z(k) = (z_1(k), z_2(k))$ under (a) $w(k)$ is independently and uniformly sampled from $[-1, 1]$, (b) $w(k) = \sin(0.1 \pi k)$. (a) and (b) show the convergence of solutions for 5 random initial conditions.}
\label{fig3}
\end{figure}
\end{example}
\subsection{Interconnected RCs for system identification} \label{subsec:sysid}
Nonlinear closed-loop model structures for black-box system identification, such as the Wiener-Hammerstein feedback models, have been proposed to better capture nonlinear feedback phenomena of the unknown system \cite{giri2010block,schoukens2017identification}. Here we introduce interconnected RCs as candidate models. When identification is entirely based on the I/O data, the closed-loop RC is required to be UOC (or UC for state-feedback interconnections, see Sec~\ref{subsec:esn-sys-id}), so that the estimated outputs for large times are determined by the inputs but not by the RC's initial condition. The internal RC parameters are arbitrary but fixed at the onset, as long as the closed-loop RC is UOC (or UC, see Sec.~\ref{subsec:esn-sys-id}). Only the RC's output function is optimized to approximate the target output data. 

Suppose that we have inputs $w_l(k), w'_{l'}(k) \in \mathbb{R}$ and their corresponding outputs $y_l(k), y'_{l'}(k) \in \mathbb{R}$ of the unknown system for $1 \leq k \leq L$, $1 \leq l \leq M$ and $1 \leq l' \leq M'$. I/O data $w_l(k), y_l(k)$ are for parameter estimation (using $l=1, \ldots, M_1$) and model selection (using $l=M_1+1, \ldots, M$, based on Akaike's final prediction error \cite{ljung1998system}). I/O data $w'_{l'}(k), y'_{l'}(k)$ are for model evaluation. For each $l$ and $l'$, we first washout the effect of RC's initial condition for $k=1,\ldots,L_w$. Let $\hat{y}_l(k), \hat{y}_{l'}(k)$ be the RC's outputs under inputs $w_l(k), w_{l'}(k)$, respectively. To optimize the RC output function, we minimize $\sum_{l=1}^{M_1} \sum_{k=L_w+1}^{L} |y_l(k) - \hat{y}_l(k)|^2$. We estimate the model order based on ${\rm FPE}_l$, computed as
$${\rm FPE}_l = \frac{1}{L-L_w} \sum_{k=L_w+1}^{L} |y_l(k) - \hat{y}_l(k)|^2 \frac{L-L_w+p}{L-L_w-p},$$
where $p$ is the number of RC output parameters. For each $p$, we randomly generate $N$ RCs and select a model out of $N p$ models with the minimum ${\rm FPE} \coloneqq \sum_{l=M_1+1}^{M} {\rm FPE}_l$.
For each $l' = 1, \ldots, M'$, the selected model is assessed using 
${\rm MSE}_{l'} = \frac{1}{L - L_w} \sum_{k=L_w+1}^{L} |y'_{l'}(k) - \hat{y}'_{l'}(k)|^2.$
We employ interconnected ESNs and QRCs to emulate the feedback-controlled Lur'e system in Sec.~\ref{subsec:obs}. We set $L_w = 500$, $L=1500$, $N=10$, $M = 10, M_1=8$ and $M' = 2$, with inputs $w_l(k)$ and $w'_1(k)$ sampled uniformly over $[-2, 2]$, independently for each $k$ (persistently exciting \cite{ljung1998system} with an order of $50$ estimated by the `pexcit' Matlab command), whereas $w'_2(k)= \sin(2 \pi k / 25) + \sin(\pi k / 5)$ as in \cite{kumpati1990identification}.

\subsubsection{Echo-state networks (ESNs)} \label{subsec:esn-sys-id}
Consider state-feedback interconnected ESNs with subsystems $j=1,2$ of the form,
\begin{equation} \label{eq:ESNs}
\begin{split}
x_j(k+1) = \tanh(A_j x_j(k) + A^{fb}_j v_j(k) + B_j w(k)), \\
\end{split}
\end{equation}
where $v_1(k) = x_2(k) \in \mathbb{R}^{n_{x_2}}, v_2(k) = x_1(k) \in \mathbb{R}^{n_{x_1}}$, $w \in l^{\infty}_1$ is the input and $\tanh(\cdot)$ is applied to a vector elementwise. We choose an output $\hat{y}(k) = W^\top_1 x_1(k) + W_2^\top x_2(k) + \zeta $, where $W_j \in \mathbb{R}^{n_{x_j}}$ and $\zeta \in \mathbb{R}$ is a bias term.  The output parameters $W_1, W_2, \zeta$ are optimized via ordinary least squares. The UC property of each ESN is guaranteed by choosing $\sigma_{\max}(A_j) < 1$ and noticing its compact state-space \cite[Theorem~13]{tran2018convergence}. We apply Corollary~\ref{corollary:uisc} to establish the UISC property for the interconnected system~\eqref{eq:ESNs}. For any $k, k_0 \in \mathbb{Z}$, let $x_j(k), \overline{x}_j(k)$ be any solutions to \eqref{eq:ESNs} under inputs $v_j, w$ and $\overline{v}_j, \overline{w}$ respectively. Let $\delta \overline{x}_j(k) = x_j(k) - \overline{x}_j(k)$, $\delta \overline{v}_j = v_j - \overline{v}_j$ and $\delta \overline{w} = w - \overline{w}$, then \eqref{eq:coro-uisc-plus} in Corollary~\ref{corollary:uisc} is satisfied since
\begin{equation*}
\begin{split}
& \| \delta \overline{x}_j(k) \| \leq \sigma_{\max}(A_j)^{k-k_0} \| \delta \overline{x}_j(k_0) \| \\
& + \frac{\sigma_{\max}(A^{fb}_j)}{1 - \sigma_{\max}(A_j)} \| \delta \overline{v}_{j_{[k_0, k-1]}} \| + \frac{\| B_j\| }{1 - \sigma_{\max}(A_j)} | \delta \overline{w}_{[k_0, k-1]}|.
\end{split}
\end{equation*}
In Corollary~\ref{corollary:uisc}, choose $\lambda_1(s) = \lambda_2(s) = \lambda s$ for some $\lambda > 0$. The closed-loop system~\eqref{eq:ESNs} is UISC if
\begin{equation} \label{eq:ESNs-SG}
\frac{\sigma_{\max}(A^{fb}_1)}{1 - \sigma_{\max}(A_1)} \frac{\sigma_{\max}(A^{fb}_2)}{1 - \sigma_{\max}(A_2)} <  \frac{1}{(1+\lambda)^2}.
\end{equation}

\begin{example}
We consider interconnected ESN~\eqref{eq:ESNs} with $n_{x_1} = n_{x_2} \in \{2, \ldots,5\}$ (i.e., $p=2n_{x_1}+1$) to model the feedback-controlled Lur'e system in Sec.~\ref{subsec:obs}. For each ESN, elements of $A_j, A^{fb}_j, B_j$ are sampled independently and uniformly over $[-1, 1]$. We fix $\sigma_{\max}(A_j) = 0.5$, $\sigma_{\max}(A^{fb}_2)$, and scale $\sigma_{\max}(A^{fb}_1)$ so that \eqref{eq:ESNs-SG} holds. The minimum FPE achieved is ${\rm FPE}=0.0023$, with $n_{x_1}=n_{x_2}=4$ and $p=9$. For this selected ESN, $\sigma_{\max}(A^{fb}_1) = 0.15, \sigma_{\max}(A^{fb}_2) = 1.65$ and \eqref{eq:ESNs-SG} holds for $\lambda=0.003$. This results in $\text{MSE}_{1}=0.0012$ and $\text{MSE}_{2}=0.0028$ corresponding to the evaluation data $l'=1, 2$, respectively. See Fig.~\ref{fig4} for the target outputs $y'_{l'}(k)$ and the closed-loop ESN outputs $\hat{y}'_{l'}(k)$. 
\begin{figure}[!ht]
\centering
\includegraphics[scale=1]{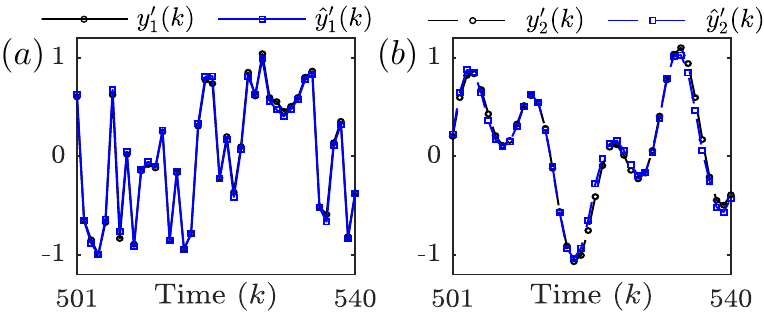}
\caption{Target outputs $y'_{l'}(k)$ and the closed-loop ESN outputs $\hat{y}'_{l'}(k)$ for $k=501, \ldots, 540$ with (a) $l'=1$ under a uniform random input $w'_1(k)$ and (b) $l'=2$ under a sum of sinusoidals $w'_2(k) = \sin(2\pi k /25) + \sin(\pi k /5)$.} 
\label{fig4}
\end{figure}
\end{example}

\subsubsection{Quantum reservoir computers (QRCs)} \label{subsec:qrc-sys-id}
We consider RCs realized by quantum dynamical systems for system identification \cite{fujii2017harnessing,CN19,chen2020temporal}. An $n$-qubit quantum system is described by a $2^n \times 2^n$ positive semidefinite Hermitian matrix $\rho$ with trace ${\rm Tr}(\rho) = 1$. A matrix $\rho$ satisfying the above properties is referred to as a density operator. We consider quantum systems evolving according to $\rho(k+1) = \mathcal{T}(w(k)) \rho(k)$, where $\mathcal{T}(w(k))$ is a completely positive trace-preserving (CPTP) map \cite{NC10} determined by input $w(k)$. A CPTP map sends a density operator to another density operator. A natural norm choice for density operators is the Schatten 1-norm, defined as $\| A \|_1 \coloneqq  {\rm Tr}(\sqrt{A^\dagger A})$ for any complex matrix $A$ and its conjugate transpose $A^\dagger$. We let $\|\mathcal{T}\|_{1-1} \coloneqq \sup_{\| A \|_1 = 1} \|\mathcal{T}(A)\|_1$ be the operator norm induced by the Schatten 1-norm. 

Consider an interconnected QRC (also see Fig.~\ref{fig5}),
\begin{equation} \label{eq:QRCs}
\begin{split}
& \begin{cases}
\rho_j(k+1) = \mathcal{T}_j(w(k), v_j(k)) \rho_j(k) + \epsilon^{(j)}_{\phi} \phi_j, \\
\hspace*{1.8em} \hat{y}_j(k) = \sum_{i=1}^{n_j} {\rm Tr}(Z_i \rho_j(k)),
\end{cases} \\
& \hspace*{5em} v_1(k) = \hat{y}_2(k), \ v_2(k) = \hat{y}_1(k),
\end{split}
\end{equation}
for $j=1,2$. Here $\mathcal{T}_j(w(k), v_j(k))= \epsilon^{(j)}_w \mathcal{T}^{(j)}_{w}(w(k)) + \epsilon^{(j)}_v \mathcal{T}^{(j)}_v (v_j(k))$, $\epsilon^{(j)}_w + \epsilon^{(j)}_v + \epsilon^{(j)}_{\phi} = 1$ and $\epsilon^{(j)}_w, \epsilon^{(j)}_v, \epsilon^{(j)}_{\phi} > 0$. Subsystem $j$ has $n_j$ qubits so that $\rho_j(k)$ and $\phi_j$ are two $2^{n_j} \times 2^{n_j}$ density operators, with $\phi_j$ being fixed. Here, $Z_i$ is the Pauli-$Z$ operator acting on qubit $i$, where Pauli-$Z$ is a $2 \times 2$ diagonal matrix with diagonal elements $1, -1$. For $S \in \{w, v\}$, the input-dependent CPTP maps are
\begin{equation*}
\mathcal{T}^{(j)}_S (x) \rho_j(k) = \left[g(x) T^{(j)}_{S, 1} + (1-g(x)) T^{(j)}_{S, 2}\right] \rho_j(k),
\end{equation*}
where $g(x) = 1 / (1+\exp(-x))$ is the logistic function with a globally Lipschitz constant $L_g = 1/4$ and $T^{(j)}_{S, 1}, T^{(j)}_{S, 2}$ are input-independent CPTP maps. We choose the output of the closed-loop QRC as 
$$\hat{y}(k) = \sum_{i=1}^{n_1} W_{i}^{(1)} {\rm Tr}(Z_i \rho_1(k)) + \sum_{i=1}^{n_2} W_{i}^{(2)} {\rm Tr}(Z_i \rho_2(k)) + \zeta,$$
where $W^{(j)}_i, \zeta \in \mathbb{R}$ ($j=1,2$ and $i=1,\ldots,n_j$) are the output parameters to be optimized via ordinary least squares.

\begin{figure}[!ht]
\centering
\includegraphics[scale=1]{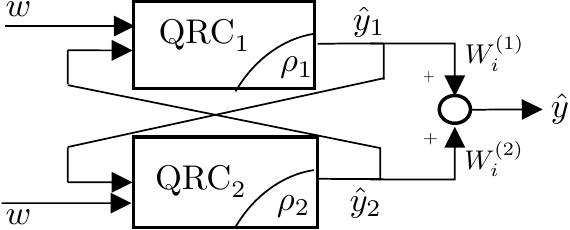}
\caption{Schematic of an interconnected QRC described by \eqref{eq:QRCs}.}
\label{fig5}
\end{figure}

Note that interconnected quantum systems do not generally take the form \eqref{eq:cl-sys}; see \cite{nurdin2009coherent}, \cite{amini2013feedback}, \cite{CKS17} and \cite[Chapter~5]{NY17}. System~\eqref{eq:QRCs} can describe ensembles of identical quantum systems such as NMR ensembles \cite{fujii2017harnessing}, and quantum systems that can emulate such ensembles; e.g., \cite{chen2020temporal,MR09}. Such quantum systems have dynamics constrained by quantum mechanics, but can otherwise be viewed as deterministic systems. Since the quantum subsystems here do not interact quantum mechanically, the composite state $\rho(k)$ for \eqref{eq:QRCs} can be described by the direct sum of the subsystem density operators, $\rho(k) = \rho_1(k) \oplus \rho_2(k)$, as for interconnected classical systems. Consequently, the closed-loop system \eqref{eq:QRCs} is of the form \eqref{eq:cl-sys} and Theorem~\ref{theorem:uioc} is applicable. We remark that Theorem~\ref{theorem:uioc} and its subsequent corollaries also hold for the Schatten 1-norm.

We now employ Corollary~\ref{corollary:uioc-plus} to establish the UIOC of system~\eqref{eq:QRCs}. Note that for any density operators $\rho, \overline{\rho}$ and any CPTP map $\mathcal{T}$, $ \| \mathcal{T}(\rho - \overline{\rho}) \|_1 \leq \| \rho - \overline{\rho}\|_1$ \cite{NC10}. For any $k, k_0 \in \mathbb{Z}$ with $k \geq k_0$, let $\rho_j(k), \overline{\rho}_j(k)$ be any solutions to inputs $v_j, w$ and $\overline{v}_j, \overline{w}$, respectively. Let $\delta \overline{\rho}_j(k) = \rho_j(k) - \overline{\rho}_j(k)$, $\delta \overline{v}_j = v_j - \overline{v}_j$ and $\delta \overline{w} = w - \overline{w}$, we have
\begin{equation} \label{eq:QRCs-bound}
\begin{split}
& \| \delta \overline{\rho}_j(k) \|_1 \leq (\epsilon^{(j)}_w + \epsilon^{(j)}_v)^{k-k_0} \| \delta \overline{\rho}_j(k_0) \|_1 \\
& \qquad + (\epsilon^{(j)}_v / \epsilon^{(j)}_{\phi}) L_g \| T^{(j)}_{v, 1} -  T^{(j)}_{v, 2} \|_{1-1} |\delta \overline{v}_{j_{[k_0, k-1]}}| \\
& \qquad + (\epsilon^{(j)}_w / \epsilon^{(j)}_{\phi}) L_g \| T^{(j)}_{w, 1} -  T^{(j)}_{w, 2} \|_{1-1} |\delta \overline{w}_{[k_0, k-1]}|.
\end{split}
\end{equation}
Furthermore, let $y_j(k), \overline{y}_j(k)$ be the outputs associated to $\rho_j(k), \overline{\rho}_j(k)$. Applying Lemma~\ref{lemma:holder-like} in Appendix~\ref{app:lemma-holder} gives 
\begin{equation} \label{eq:QRCs-bound2}
| y_j(k) - \overline{y}_j(k) | = \left| \sum_{i=1}^{n_j} {\rm Tr}(Z_i \delta \overline{\rho}_j(k) ) \right| \leq n_j \| \delta \overline{\rho}_j(k) \|_1.
\end{equation}

To show that each QRC subsystem is UOC, note that a quantum system admits a compact state-space. Consider $\delta \overline{v}_j = \delta \overline{w}_j = 0$. From \eqref{eq:QRCs-bound}, we have $\| \delta \overline{\rho}_j(k) \| \leq (\epsilon^{(j)}_w + \epsilon^{(j)}_v)^{k-k_0} \| \delta \overline{\rho}_j(k_0) \|_1$ with $\epsilon^{(j)}_w + \epsilon^{(j)}_v < 1$. By \cite[Theorem~13]{tran2018convergence}, there exists a unique bounded reference state solution $\rho^*_j$ to each subsystem in \eqref{eq:QRCs}.
From \eqref{eq:QRCs-bound} and \eqref{eq:QRCs-bound2}, we have $| y_j(k) - \overline{y}_j(k) | \leq n_j (\epsilon^{(j)}_w + \epsilon^{(j)}_v)^{k-k_0} \| \delta \overline{\rho}_j(k_0) \|_1$, and hence UOC of each QRC subsystem.

Upper bounding $\| T^{(j)}_{S, 1} -  T^{(j)}_{S, 2} \|_{1-1} \leq 2$ for $S\in \{w, v\}$ \cite[Theorem~2.1]{perez2006contractivity} in \eqref{eq:QRCs-bound}. Equations~\eqref{eq:QRCs-bound}, \eqref{eq:QRCs-bound2} show that \eqref{eq:theorem-uioc-x}, \eqref{eq:alternative-uioc} in Corollary~\ref{corollary:uioc-plus} hold. Choose $\lambda_1(s) = \lambda_2(s) = \lambda s $ for some $\lambda > 0$ in Corollary~\ref{corollary:uioc-plus}. From \eqref{eq:QRCs-bound}, \eqref{eq:QRCs-bound2}, the closed-loop QRC~\eqref{eq:QRCs} is UIOC if
\begin{equation} \label{eq:QRC-SG}
(4 \epsilon^{(1)}_v \epsilon^{(2)}_v L_g L_g n_1 n_2) / (\epsilon^{(1)}_{\phi} \epsilon^{(2)}_{\phi}) < 1/(1+\lambda)^2.
\end{equation}

\begin{example}
We consider an interconnected QRC \eqref{eq:QRCs} with $n_1 = n_2 \in \{2, \ldots, 5\}$ (i.e., $p=2n_1+1$) to model the feedback-controlled Lur'e system in Sec.~\ref{subsec:obs}. For each QRC, we fix $\epsilon^{(1)}_w = 0.25, \epsilon^{(1)}_v = 0.1, \epsilon^{(1)}_{\phi} = 0.65$ and $\epsilon^{(2)}_w = 0.1, \epsilon^{(2)}_v = 0.45, \epsilon^{(2)}_{\phi} = 0.45$, such that \eqref{eq:QRC-SG} holds for all values of $n_1$ considered here. For $j=1,2$, $\phi_j$ is chosen with its $1,1$-th element $(\phi_j)_{1,1} = 1$ and zero otherwise. Each input-independent CPTP map is governed by a unitary matrix $U^{(j)}_{S, m}$, defined by $T^{(j)}_{S, m}(\rho) = U^{(j)}_{S, m} \rho (U^{(j)}_{S, m})^\dagger$ for $j, m=1, 2$ and $S \in \{w, v\}$. More explicitly, we choose $U^{(1)}_{w, 1} = U^{(1)}_{v, 2} = U^{(2)}_{w, 1} = U^{(2)}_{v, 2} = \bigotimes_{i=1}^{n_1} Z_i$, where $\otimes$ denotes the tensor (Kronecker) product. Other unitaries are $U^{(j)}_{S, m} = \bigotimes_{i=1}^{n_1} e^{-\iota \theta^{S, m, j}_{i} X_i}$, where $\iota = \sqrt{-1}$, $X_i$ is the Pauli-$X$ operator on qubit $i$ and the Pauli-$X$ operator is $\begin{bmatrix} 0 & 1 \\ 1 & 0 \end{bmatrix}$. Parameters $\theta^{S, m, j}_{i}$ are uniformly distributed on $[-\pi, \pi]$, independently for each $S, m, j$ and $i$. The unitaries employed here are simple, more complex unitaries that entangle qubits within a QRC subsystem can also be used; see \cite{fujii2017harnessing,CN19,chen2020temporal}.

The minimum FPE is achieved at ${\rm FPE} = 0.0032$ with $n_1=n_2=5$ and $p=11$, and \eqref{eq:QRC-SG} holds for $\lambda=0.019$. This selected QRC achieves ${\rm MSE}_1={\rm MSE}_2=0.0015$. See Fig.~\ref{fig6} for the QRC outputs $\hat{y}'_{l'}(k)$ against the target outputs $y'_{l'}(k)$ for the evaluation data $l'=1, 2$. 
\begin{figure}[!ht]
\centering
\includegraphics[scale=1]{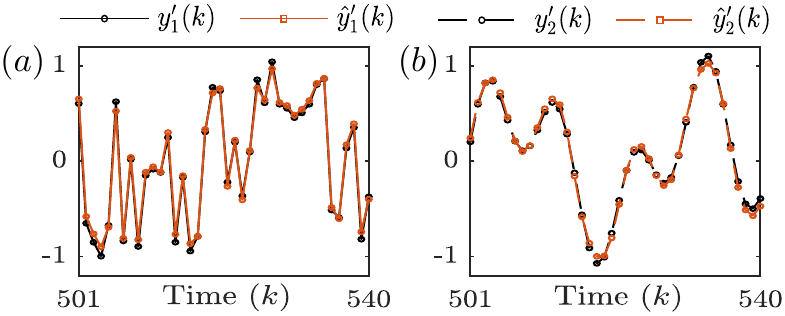}
\caption{Target outputs $y'_{l'}(k)$ and the closed-loop QRC outputs $\hat{y}'_{l'}(k)$ for $k=501, \ldots, 540$ with (a) $l'=1$ under a uniform random input $w'_1(k)$ and (b) $l'=2$ under a sum of sinusoidals $w'_2(k) = \sin(2\pi k /25) + \sin(\pi k /5)$.} 
\label{fig6}
\end{figure}
\end{example}

\section{Conclusion}
\label{sec:conclusion}
We present a small-gain theorem for output-feedback interconnected systems to be uniformly input-to-output convergent systems, as a discrete-time counterpart of the continuous-time results in \cite{besselink2011model}. Our proof is based on a small-gain theorem for time-varying discrete-time systems in the input-to-output stability framework, also derived herein. The latter result bridges the gap between time-invariant and time-varying discrete-time small-gain theorems in the literature \cite{jiang2001input,jiang2005remarks}.

Our small-gain theorems are applicable to important control problems, such as output regulation and tracking \cite{pavlov2006uniform}.  We demonstrate an application of our small-gain theorems to observer-based controller design, illustrated with systems subject to globally Lipschitz nonlinearities that are ubiquitous in mechanical and robotic applications. Detouring from conventional applications, we apply the uniform input-to-output convergence small-gain theorem to design parameters of interconnected reservoir computers for black-box system identification. We introduce interconnected echo-state networks and quantum reservoir computers as candidate models equipped with closed-loop structures and demonstrate numerically their efficacy in modeling a feedback-controlled system.

\appendices
\section{}
\label{app:theorem_uios} 
To prove Theorem~\ref{theorem:uios}, we will apply Lemma~\ref{lemma:ex_beta} below.
\begin{lemma} \label{lemma:ex_beta}
Consider a well-posed system \eqref{eq:cl-sys}. Suppose that for $j=1,2$, there exists $\tilde{\beta}_j \in \mathcal{KL}$, $\tilde{\gamma}^u_j, \tilde{\gamma}^y_j \in \mathcal{K}$ (independent of $u_j, v_j$) with $\tilde{\gamma}^y_j(s) < s$ for all $s>0$ such that,for some $M \in \mathbb{Z}$ with $M \geq 2$, for any input $u \in l^\infty_{n_{u_1}+ n_{u_2}}$, any $k_0 \in \mathbb{Z}$ and $k_1 \in \mathbb{Z}_{\geq 0}$, and any $x(k_0) \in \mathbb{R}^{n_{x_1} + n_{x_2}}$,
\begin{equation} \label{eq:lemma_ex_beta}
\begin{split}
\| y_j (k_0 & + k_1) \|  \leq \max \{ \tilde{\beta}_j( \| x(k_0) \|, k_1), \\
& \tilde{\gamma}^y_j(\| y_{j_{[k_0 + \lfloor k_1 / M \rfloor, k_0 + k_1]}} \|), \tilde{\gamma}^u_j(\| u_{[k_0, k_0 + k_1]}\|) \},
\end{split}
\end{equation}
and 
\begin{equation*}
\sup_{k_1 \in \mathbb{Z}_{\geq 0}} \| y_j(k_0 + k_1) \| < \infty.
\end{equation*}
Then there exists $\hat{\beta}_j \in \mathcal{KL}$ such that, for all $k_0 \in \mathbb{Z}$, $k_1 \in \mathbb{Z}_{\geq 0}$,
\begin{equation*}
\| y_j (k_0 + k_1) \|  \leq \max \{ \hat{\beta}_j( \| x(k_0) \|, k_1),  \tilde{\gamma}^u_j(\| u_{[k_0, k_0 + k_1]}\|) \}.
\end{equation*}
\end{lemma}

The main idea in the proof of Lemma~\ref{lemma:ex_beta} is to apply a continuous extension argument and \cite[Lemma~A.2]{sontag2002small} stated below in Lemma~\ref{lemma:sontag}.
\begin{lemma} \label{lemma:sontag} 
Given $\delta \in \mathcal{K}_{\infty}$ and $T:(0, \infty) \times (0, \infty) \rightarrow \mathbb{R}_{\geq 0}$ such that, (i) for all $\epsilon > 0$, $s_1 < s_2$ implies $T(\epsilon, s_1) \leq T(\epsilon, s_2)$; (ii) for all $s > 0$, $\lim_{\epsilon \rightarrow 0^+} T(\epsilon, s) = \infty$. Then there exists $\hat{\beta} \in \mathcal{KL}$ such that, for each $s >0$ and $t_1 \in \mathbb{R}_{\geq 0}$, there exists some $\epsilon \in A_{s, t_1} \coloneqq \{\epsilon' \in (0, \infty) \mid t_1 \geq T(\epsilon', s)\} \cup \{\infty\}$ such that $\min\{\epsilon, \delta^{-1}(s)\} \leq \hat{\beta}(s, t_1)$.
\end{lemma}

\begin{proof}[Proof of Lemma~\ref{lemma:ex_beta}]
Fix $k_0 \in \mathbb{Z}$. For any $k_1 \in \mathbb{Z}_{\geq 0}$, define $z_j(k_1) = \| y_j(k_0 + k_1) \| $ if $\| y_j(k_0 + k_1) \| > \tilde{\gamma}^u_j(\| u_{[k_0, k_0 + k_1]}\|)$ and $z_j(k_1) = 0$ otherwise. From the assumption \eqref{eq:lemma_ex_beta} in Lemma~\ref{lemma:ex_beta}, we have that for some $M \in \mathbb{Z}, M \geq 2$,
\begin{equation} \label{eq:lemma-proof1}
z_j(k_1) \leq \max \{ \tilde{\beta}_j(\|x(k_0)\|, k_1), \tilde{\gamma}^y_j(|z_{j_{[ \lfloor k_1/M \rfloor, k_1]}}|) \}.
\end{equation}
Note the implicit dependence $z_j(k_1) = z_j(k_1, x(k_0), u)$.
We sample and hold the left points to extend $z_j$ to a piecewise continuous function $w_j$. For any $t_1 \in \mathbb{R}_{\geq 0}$, define $w_j(t_1) = \sum_{k_1 = 0}^{\infty} \mathds{1}_{[k_1, k_1+1)}(t_1) z_j(k_1)$, where $\mathds{1}_{[k_1, k_1+1)}(t_1) = 1$ if $t_1 \in [k_1, k_1 + 1)$ and zero otherwise. For any $\tau, \tau' \in \mathbb{R}$, let $|w_{j_{[\tau, \tau']}}| \coloneqq \sup_{\tau \leq \overline{\tau} \leq \tau'} |w_j(\overline{\tau})|$. Since $\lfloor t_1 \rfloor \leq t_1 < \lfloor t_1 \rfloor + 1$, we have $w_j(t_1) = z_j(\lfloor t_1 \rfloor) = w_j(\lfloor t_1 \rfloor)$ and
\begin{equation} \label{eq:lemma-proof1-1}
\begin{split}
|w_{j_{[t_1 / M, t_1]}}| = |w_{j_{[\lfloor t_1/M \rfloor, \lfloor t_1 \rfloor]}}| & = |w_{j_{[\lfloor \lfloor t_1 \rfloor /M \rfloor, \lfloor t_1 \rfloor]}}| \\
& = |z_{j_{[\lfloor \lfloor t_1 \rfloor /M \rfloor, \lfloor t_1 \rfloor]}}|. 
\end{split}
\end{equation}
For $s \in \mathbb{R}_{\geq 0}$, let $\beta_j(s, t_1)= \sum_{k_1=0}^{\infty} \mathds{1}_{[k_1, k_1+1)}(t_1) \tilde{\beta}_j(s, k_1)$. Then $\beta_j \in \mathcal{KL}$ and $\beta_j(s, t_1) = \tilde{\beta}_j(s, \lfloor t_1 \rfloor)$. From \eqref{eq:lemma-proof1} and \eqref{eq:lemma-proof1-1}, we have that for all $t_1 \in \mathbb{R}_{\geq 0}$,
\begin{equation} \label{eq:lemma-proof2}
w_j(t_1) \leq \max\{ \beta_j(\| x(k_0)\|, t_1), \tilde{\gamma}^y_j(|w_{j_{[t_1/ M, t_1]}}|)\}.
\end{equation}
To apply Lemma~\ref{lemma:sontag}, we first show the following claims.

\textit{Claim (i): There exists $\delta \in \mathcal{K}_\infty$ such that for all $t_1 \in \mathbb{R}_{\geq 0}$ and $x(k_0) \in \mathbb{R}^{n_{x_1} + n_{x_2}}$, we have $w_j(t_1) \leq \delta^{-1}(\| x(k_0)\|)$.} 

\begin{proof}[Proof of Claim (i)]
Note that $|w_{j_{[0, \infty)}}|  \coloneqq \sup_{t_1  \in \mathbb{R}_{\geq 0}} w_j(t_1)$ $= \sup_{k_1 \in \mathbb{Z}_{\geq 0}} \| y_j(k_0 + k_1) \| < \infty$. From \eqref{eq:lemma-proof2}, we have
\vspace*{-0.4em}
\begin{equation*}
\vspace*{-0.4em}
|w_{j_{[0, \infty)}}| \leq \max\{\beta_j(\|x(k_0)\|, 0), \tilde{\gamma}^y_j(|w_{j_{[0, \infty)}}|)\}.
\end{equation*}
Since $\tilde{\gamma}^y_j(s) < s$ for all $s > 0$, it follows that for any $t_1 \in \mathbb{R}_{\geq 0}$,
\begin{equation} \label{eq:claim1}
w_j(t_1) \leq |w_{j_{[0, \infty)}}| \leq \beta_j(\| x(k_0) \|, 0).
\end{equation}
Choose $\delta \in \mathcal{K}_\infty$ such that $\delta^{-1}(\|x(k_0)\|) \geq \beta_j(\|x(k_0)\|, 0)$ (e.g., $\delta^{-1} = id + \beta_j(\cdot, 0) $) gives the desired result.
\end{proof}

\textit{Claim (ii): For any $\epsilon, r > 0$, there exists $\hat{T}_{\epsilon, r} \in \mathbb{R}_{\geq 0}$ such that for all $t_1 \geq \hat{T}_{\epsilon, r}$, $w_j(t_1) \leq \epsilon$ whenever $\|x(k_0)\| \leq r$.}

\begin{proof}[Proof of Claim (ii)]
The proof uses \eqref{eq:claim1} and proceeds as in \cite[Lemma~2.1]{chen2005simplified}. Let $\epsilon, r > 0$, if $\beta_j(\| x(k_0) \|, 0) \leq \beta_j(r, 0) \leq \epsilon$, then by \eqref{eq:claim1}, $w_j(t_1) \leq  \beta_j(\|x(k_0)\|, 0) \leq \epsilon$ for all $t_1 \in \mathbb{R}_{\geq 0}$. Otherwise, since $\tilde{\gamma}^y_j$ is strictly contractive, there exists $n_{\epsilon, r} \in \mathbb{Z}_{\geq 0}$ such that the $n_{\epsilon, r}$-times composition $(\tilde{\gamma}^y_j)^{(n_{\epsilon, r})}(\beta_j(r, 0)) \leq \epsilon$. For $i=1, \ldots, n_{\epsilon, r}$, let $\tau_i \in \mathbb{R}_{\geq 0}$ be the first time instance such that $\beta_j(r, \tau_i)\leq  (\tilde{\gamma}^y_j)^{(i)}(\beta_j(r, 0))$ so that $\tau_i \leq \tau_j$ for $j=i+1, \ldots, n_{\epsilon,r}$. Define $\hat{\tau}_0 = 0$ and $\hat{\tau}_i = \max\{\tau_i, M\hat{\tau}_{i-1} \}$. We will show by induction that for $t_1 \geq \hat{\tau}_i$, $w_j(t_1) \leq (\tilde{\gamma}^y_j)^{(i)}(\beta_j(r, 0))$. 

Claim (i) establishes the case for $i=0$ (with $(\tilde{\gamma}^{y}_j)^{(0)} = id$). Suppose the induction hypothesis holds for $t_1 \geq \hat{\tau}_i$. For $t_1 \geq \hat{\tau}_{i+1}$, we have $t_1 \geq \tau_{i+1}$ and $t_1 / M \geq \hat{\tau}_{i}$. From \eqref{eq:lemma-proof2},
\begin{equation*}
\begin{split}
w_j(t_1) & \leq \max\{ \beta_j(\|x(k_0)\|, \tau_{i+1}), \tilde{\gamma}^y_j \circ (\tilde{\gamma}^y_j)^{(i)}(\beta_j(r, 0)) \} \\
& = (\tilde{\gamma}^y_j)^{(i+1)}(\beta_j(r, 0)).
\end{split}
\end{equation*}
Claim~(ii) follows from choosing $\hat{T}_{\epsilon, r} \geq \hat{\tau}_{n_{\epsilon ,r}}$.
\end{proof}

Let $\hat{T}_{\epsilon, r}$ be given by Claim (ii). As in \cite[Proposition~2.7]{sontag2002small}, define $T(\epsilon, r) = r/\epsilon + \inf\{\hat{T}_{\epsilon', r'} | r \leq r', \epsilon' \in (0, \epsilon] \}$. Then $T(\cdot, \cdot)$ satisfies the conditions in Lemma~\ref{lemma:sontag}. Fix $s = \| x(k_0) \| > 0$ (the case for $s=0$ is immediate), any $t_1 \in \mathbb{R}_{\geq 0}$ and the set $A_{s, t_1}$. By Claim (ii), $w_j(t_1) \leq \epsilon$ for all $\epsilon \in A_{s, t_1}$. Let $\hat{\beta}_j \in \mathcal{KL}$ and $\epsilon \in A_{s, t_1}$ be given by Lemma~\ref{lemma:sontag}, such that $\min\{ \epsilon , \delta^{-1}(\|x(k_0)\|) \} \leq \hat{\beta}_j(\|x(k_0)\|, t_1)$. Then
\begin{itemize}
\item If $\delta^{-1}(\|x(k_0)\|) \leq \epsilon$, then by Claim (i) we have $w_j(t_1) \leq \delta^{-1}(\|x(k_0)\|) \leq \hat{\beta}_j(\|x(k_0)\|, t_1)$. 
\item If $\epsilon < \delta^{-1}(\|x(k_0)\|)$, then $\epsilon < \infty$ and $w_j(t_1) \leq \epsilon \leq \hat{\beta}_j(\|x(k_0)\|, t_1)$.
\end{itemize}
Therefore, for all $t_1 \in \mathbb{R}_{\geq 0}$, $w_j(t_1) \leq \hat{\beta}_j(\|x(k_0)\|, t_1)$. In particular, for all $k_1 \in \mathbb{Z}_{\geq 0}$, $w_j(k_1) = z_j(k_1)  \leq \hat{\beta}_j(\|x(k_0)\|, k_1)$. By definition of $z_j(k_1)$, we have the desired result.
\end{proof}

We now prove Theorem~\ref{theorem:uios}. The proof adapts \cite[Theorem~2.1]{chen2005simplified} to discrete-time systems of the form \eqref{eq:cl-sys}. We first show that $\sup_{k \geq k_0} \| x(k) \| < \infty$ and $\sup_{k \geq k_0} \| y(k) \| < \infty$, then we apply Lemma~\ref{lemma:ex_beta} with $M=4$ to show that system~\eqref{eq:cl-sys} is UIOS.

\begin{proof}[Proof of Theorem~\ref{theorem:uios}]

From \eqref{eq:theorem-uios-y}, we have that for all $k \geq k_0$ and $j= 1, 2$,
\begin{equation*}
\begin{split}
\| y_{j_{[k_0, k]}} \| \leq \max \{ & \beta_j(\| x_j(k_0) \|, 0), \\
& \gamma^y_j(\|v_{j_{[k_0, k]}} \|),  \gamma^u_j(\|u_{j_{[k_0, k]}} \|)\}. 
\end{split}
\end{equation*}
Substituting $v_1 = y_2$,  $v_2 = y_1$ and the bound for $\| y_{2_{[k_0, k]}} \|$ into that of $\| y_{1_{[k_0, k]}} \|$, we have
\begin{equation} \label{eq:uios-bound-1}
\begin{split}
& \| y_{1_{[k_0, k]}} \| \\
& \leq \max\{  \beta_1(\| x_1(k_0) \|, 0), \gamma_1^y \circ \beta_2(\|x_2(k_0)\|, 0), \\ 
& \gamma^y_1\circ\gamma^y_2(\| y_{1_{[k_0, k]}}\|), \gamma_1^y \circ \gamma_2^u(\| u_{2_{[k_0, k]}} \|), \gamma_1^u(\|u_{1_{[k_0, k]}}\|) \} \\
& \leq \max\{  \beta_1(\| x_1(k_0) \|, 0), \gamma_1^y \circ \beta_2(\|x_2(k_0)\|, 0), \\ 
& \hspace*{4em} \gamma_1^y \circ \gamma_2^u(\| u_{2_{[k_0, k]}} \|), \gamma_1^u(\|u_{1_{[k_0, k]}}\|) \},
\end{split}
\end{equation}
where the last inequality follows from $\gamma_1^y \circ \gamma_2^y(\| y_{1_{[k_0, k]}} \|) < \| y_{1_{[k_0, k]}} \|$. A symmetric argument shows that 
\begin{equation} \label{eq:uios-bound-2}
\begin{split}
& \hspace*{-0.8em} \| y_{2_{[k_0, k]}} \| \leq \max\{  \beta_2(\| x_2(k_0) \|, 0) , \gamma_2^y \circ \beta_1(\|x_1(k_0)\|, 0),\\ 
& \hspace*{6.5em} \gamma_2^y \circ \gamma_1^u(\| u_{1_{[k_0, k]}} \|), \gamma_2^u(\|u_{2_{[k_0, k]}}\|) \}.
\end{split}
\end{equation}
Recall that $u_j \in l^{\infty}_{n_{u_j}}$ for $j=1,2$. From \eqref{eq:uios-bound-1} and \eqref{eq:uios-bound-2}, $\sup_{k \geq k_0} \| y_j(k) \| < \infty$. Substituting $v_{1} = y_{2}$ and $v_{2} =  y_{1}$ in \eqref{eq:theorem-uios-x}, it follows that $\sup_{k \geq k_0} \| x_j(k) \| < \infty$. It remains to show that system~\eqref{eq:cl-sys} is UIOS.

Upper bound $\|v_{j_{[k_0, k-1]}}\|$ in \eqref{eq:theorem-uios-x} by \eqref{eq:uios-bound-1} and \eqref{eq:uios-bound-2},  using $\| x_j(k_0) \| \leq \|x(k_0)\|$ and $\| u_{j_{[k_0, k]}}\| \leq  \| u_{[k_0, k]}\|$, we have
\begin{equation} \label{eq:x-bound}
\begin{split}
& \hspace*{-1em}  \| x(k) \| \leq 2 \max \{\| x_1(k) \| , \| x_2(k) \| \}  \\
& \hspace*{-1em} \leq 2 \max_{j=1,2} \left\{ \sigma_j ( \left\| x(k_0) \right\| ), \sigma_j^y ( \| v_{j_{[k_0, k]}} \| ), \sigma^u_j ( \left\| u_{[k_0, k]} \right\| ) \right\} \\
& \hspace*{-1em}  \leq \max \left\{ \overline{\sigma}( \left\| x(k_0) \right\| ), \overline{\gamma}( \left\| u_{[k_0, k]} \right\| )  \right\},
\end{split}
\end{equation}
where $\overline{\sigma}(s) = 2 \max \{ \sigma_1(s), \sigma_2(s),  \sigma_1^y (\beta_2(s, 0)),\sigma_2^y (\beta_1(s, 0)),$ $\sigma_1^y \circ \gamma_2^y (\beta_1(s, 0)), \sigma_2^y \circ \gamma_1^y (\beta_2(s, 0)) \}$ and $\overline{\gamma}(s) = 2 \max \{ \sigma^u_1(s),$  $\sigma^u_2(s), \sigma_1^y \circ \gamma_2^y \circ \gamma_1^u(s), \sigma_2^y \circ \gamma_1^y \circ \gamma_2^u(s), \sigma_1^y \circ\gamma^u_2(s), \sigma_2^y \circ \gamma^u_1(s) \}$.

Consider subsystem $j=1$ and \eqref{eq:theorem-uios-y}. For any $k_1 \in \mathbb{Z}_{\geq 0}$ and $k_0 \in \mathbb{Z}$, let $k_0 + \lceil k_1 / 2 \rceil$ be the initial time and $k = k_0 + k_1$. Then $k - (k_0 +  \lceil k_1 / 2 \rceil) = \lfloor k_1 / 2 \rfloor$ and
\begin{equation}
\label{eq:uios-proof1}
\begin{split}
& \| y_{1}(k_0+  k_1) \|  \\
&  \leq \max \{\beta_1(\|x_1(k_0 + \lceil k_1 / 2 \rceil) \|, \lfloor k_1 / 2 \rfloor), \\ 
&  \gamma^y_1(\| y_{2_{[k_0 + \lceil k_1/2 \rceil, k_0 + k_1]}} \|), \gamma^u_1(\|u_{1_{[k_0 + \lceil k_1 / 2 \rceil, k_0 + k_1]}} \|) \} \\
& \leq \max \{\beta_1(\|x(k_0 + \lceil k_1 / 2 \rceil) \|, \lfloor k_1 / 2 \rfloor), \\ 
&  \gamma^y_1(\| y_{2_{[k_0 + \lceil k_1/2 \rceil, k_0 + k_1]}} \|), \gamma^u_1(\|u_{{[k_0, k_0 + k_1]}} \|) \}.
\end{split}
\end{equation}
Consider subsystem $j=2$ and \eqref{eq:theorem-uios-y}. Let $k_0 + \lfloor k_1/ 4 \rfloor$ be the initial time. For any $\lceil k_1/2 \rceil \leq \overline{k}_1 \leq k_1$, let $k = k_0 + \overline{k}_1$. Then $k- (k_0 +  \lfloor k_1 / 4 \rfloor) = \overline{k}_1 -  \lfloor k_1 / 4 \rfloor \geq \lceil k_1/2 \rceil - \lfloor k_1/4 \rfloor \geq \lceil k_1/4 \rceil$, 
\begin{equation}\label{eq:uios-proof2}
\begin{split}
& \left\| y_{2} (k_0 + \overline{k}_1) \right\| \\
& \leq \max \{ \beta_2(\|x_2(k_0 + \lfloor k_1 / 4 \rfloor ) \|,  \overline{k}_1 -   \lfloor k_1/4 \rfloor), \\
& \gamma_2^y(\| y_{1_{[k_0 + \lfloor k_1 / 4 \rfloor, k_0 + \overline{k}_1]}}\| ), \gamma_2^u(\| u_{2_{[k_0+ \lfloor k_1 / 4 \rfloor, k_0 + \overline{k}_1]}} \| ) \} \\
& \leq \max \{ \beta_2(\|x(k_0 + \lfloor k_1 / 4 \rfloor ) \|, \lceil k_1/4 \rceil), \\
&   \gamma_2^y(\| y_{1_{[k_0 + \lfloor k_1 / 4 \rfloor, k_0 + k_1]}}\| ), \gamma_2^u(\| u_{{[k_0, k_0 + k_1]}} \| ) \}
\end{split}
\end{equation}
Note that the right-hand side of the last inequality in \eqref{eq:uios-proof2} does not depend on $\overline{k}_1$. Now taking $\sup_{\lceil k_1/2 \rceil \leq \overline{k}_1 \leq k_1}$ on both sides of \eqref{eq:uios-proof2} shows that 
\begin{equation} \label{eq:uios-proof3}
\begin{split}
& \| y_{2_{[k_0 + \lceil k_1/2 \rceil, k_0 + k_1]}} \| \\
& \leq \max \{ \beta_2(\|x(k_0 + \lfloor k_1 / 4 \rfloor ) \|, \lceil k_1/4 \rceil), \\ 
& \hspace*{4em} \gamma_2^y(\| y_{1_{[k_0 + \lfloor k_1 / 4 \rfloor, k_0 + k_1]}}\| ), \gamma_2^u(\| u_{{[k_0, k_0 + k_1]}} \| ) \}.
\end{split}
\end{equation}
Upper bounding $\| y_{2_{[k_0 + \lceil k_1/2 \rceil, k_0 + k_1]}} \|$ in \eqref{eq:uios-proof1} using \eqref{eq:uios-proof3} and upper bounding $\| x(k_0 + \lceil k_1/2 \rceil) \|$ in \eqref{eq:uios-proof1} using \eqref{eq:x-bound}, we have
\begin{equation} \label{eq:uios-proof4} 
\begin{split}
\| y_1(k_0 + k_1) \| \leq \max \{ \beta_1( \overline{\sigma}(\| x(k_0) \|), \lfloor k_1/2 \rfloor),& \\
\gamma_1^y \circ \beta_2(\overline{\sigma}(\| x(k_0) \| ), \lceil k_1/4 \rceil), &\\
\gamma_1^y \circ \gamma_2^y(\| y_{1_{[k_0 + \lfloor k_1 / 4 \rfloor, k_0 + k_1]}} \|),  \gamma_1^y \circ \gamma_2^u(\|u_{{[k_0, k_0 + k_1]}} \|), & \\
\gamma_1^u(\| u_{{[k_0, k_0 + k_1]}} \|), \beta_1 ( \overline{\gamma}(\| u_{[k_0, k_0+k_1]} \|), 0), & \\
\gamma_1^y \circ \beta_2(\overline{\gamma}(\| u_{[k_0, k_0+k_1]} \|), 0)  \}& .
\end{split}
\end{equation}
Let $\tilde{\gamma}^y_1(s) = \gamma_1^y \circ \gamma_2^y(s)$, $\tilde{\beta}_1(s, k_1) = \max\{ \beta_1(\overline{\sigma}(s), \lfloor k_1/ 4 \rfloor),$ $\gamma^y_1 \circ \beta_2(\overline{\sigma}(s),  \lfloor k_1/ 4 \rfloor) \}$ and $\tilde{\gamma}^u_1(s) = \max \{ \gamma_1^y \circ \gamma_2^u(s), \gamma_1^u(s),$ $\beta_1(\overline{\gamma}(s), 0), \gamma_1^y \circ \beta_2(\overline{\gamma}(s), 0)\}$. By a symmetric argument, we can define $\tilde{\beta}_2$ and $\tilde{\gamma}^u_2$ analogously. Here, $\tilde{\beta}_j \in \mathcal{KL}$ and $\tilde{\gamma}^u_j \in \mathcal{K}$ for $j=1, 2$. Re-writing \eqref{eq:uios-proof4} in terms of $\tilde{\beta}_j$ and $\tilde{\gamma}_j$, we have
\begin{equation*}
\begin{split}
\| y_j (k_0 + k_1) \| &\leq \max \{ \tilde{\beta}_j(\|x(k_0)\|, k_1), \\
& \hspace*{1.5em} \tilde{\gamma}^y_j(\| y_{j_{[k_0 + \lfloor k_1/4 \rfloor, k_0 + k_1]}} \|), \tilde{\gamma}^u_j(\| u_{[k_0, k_0 + k_1]}\|)\},
\end{split}
\end{equation*}
where $\tilde{\gamma}^y_j(s) < s$ for all $s > 0$ by strict contractivity of $\gamma_1^y \circ \gamma^y_2(\cdot)$ and $\gamma_2^y \circ \gamma^y_1(\cdot)$, and we have shown $ \sup_{k_1 \in \mathbb{Z}_{\geq 0}} \| y_j (k_0 + k_1) \| < \infty$.
Invoking Lemma~\ref{lemma:ex_beta} with $M=4$, there exists $\hat{\beta}_j \in \mathcal{KL}$ such that for all $k_0 \in \mathbb{Z}$ and $k_1 \in \mathbb{Z}_{\geq 0}$,
\begin{equation*}
\| y_j(k_0 + k_1) \| \leq \max \{ \hat{\beta}_j(\| x(k_0) \|, k_1), \tilde{\gamma}^u_j(\| u_{[k_0, k_0 + k_1]}\|)\}.
\end{equation*}
Let $k=k_0+k_1$, $\gamma(s) = 2 \max_{j=1, 2}\{ \tilde{\gamma}^u_j(s)\}$ and $\beta(s, k) = 2\max_{j=1, 2} \{ \tilde{\beta}_j(s, k)\}$. It follows that
\begin{equation*}
\begin{split}
\| y(k) \| 
& \leq 2 \max \{\| y_1(k) \|, \| y_2(k) \| \} \\
& \leq \max \{ \beta(\| x(k_0)\|, k-k_0), \gamma(\| u_{[k_0, k]}\|) \}
\end{split}
\end{equation*} 
for all $k, k_0 \in \mathbb{Z}$ with $k \geq k_0$ and any initial condition $x(k_0) \in \mathbb{R}^{n_{x_1} + n_{x_2}}$.
\end{proof}

\section{}
\label{app:obs_coro}
\begin{proof}[Proof of Corollary~\ref{corollary:obs}]
By \cite[Remark~5]{tran2018convergence}, for any inputs $v_2, u_2$, $x^*_2=0$ is the unique and bounded reference state solution, so that subsystem $j=2$ is UC. To show closed-loop UISC, let $\overline{x}_2(k)$ be any solution to inputs $\overline{v}_2, \overline{u}_2$ with initial condition $\overline{x}_2(k_0)$. Equation~\eqref{eq:coro-uioc-obs} implies that for any gains $\gamma^y_2, \gamma^u_2 \in \mathcal{K}$,
\begin{equation*}
\begin{split}
& \| x_2^*(k) - \overline{x}_2(k) \| \leq \beta_2(\| x_2^*(k_0) - \overline{x}_2(k_0) \|, k - k_0) \\
& \hspace*{1em} + \gamma^y_2(\| (v_2 - \overline{v}_2)_{[k_0, k-1]} \|) + \gamma^u_2(\| (u_2 - \overline{u}_2)_{[k_0, k-1]} \|) \}.
\end{split}
\end{equation*}
The above equation and \eqref{eq:coro-uioc-obs-x1} shows that \eqref{eq:coro-uisc-plus} in Corollary~\ref{corollary:uisc} hold. We now show that \eqref{eq:coro-uios-plus} in Corollary~\ref{corollary:uisc} is satisfied and hence the closed-loop UISC of system~\eqref{eq:app-cl}. 
Observe that $id + \gamma_1^y \in \mathcal{K}_{\infty}$ with $(id + \gamma_1^y)^{-1} \circ \gamma_1^y(s) < s$ for all $s > 0$. Let $\lambda_1, \lambda_2 \in \mathcal{K}_{\infty}$ be arbitrary. The result follows from choosing $\gamma_2^y(s) = ((id+\lambda_1) \circ (id+\gamma^y_1) \circ (id + \lambda_2))^{-1}(s)$.
\end{proof}

\section{}
\label{app:obs_lmi}
We show that if the linear matrix inequalities \eqref{eq:lmi} are satisfied then \eqref{eq:coro-uioc-obs} holds for the observer error dynamics. 

\begin{lemma} \label{lemma:obs_lmi}
Consider the observer error dynamics $\Delta z(k+1) = (A - P^{-1} Z C) \Delta z(k) + \rho G \sin(Hz(k)) - \rho G \sin(H \hat{z}(k))$. Suppose that there exists $Z \in \mathbb{R}^2$, $P > 0$ and $\epsilon > 0$ such that 
\begin{small}
\begin{equation}
\begin{split}
P - \epsilon I & < 0, \\
\begin{bmatrix}
-\theta P & A^\top P - C^\top Z^\top & \epsilon \rho (GH)^\top & A^\top P - C^\top Z^\top \\
PA - ZC & -P & 0 & 0 \\
\epsilon \rho GH & 0 & -\epsilon I & 0 \\
PA - ZC & 0 & 0 & P - \epsilon I
\end{bmatrix} & \leq 0.
\end{split}
\end{equation}
\end{small}
\noindent Then the observer error dynamics satisfies \eqref{eq:coro-uioc-obs} in Corollary~\ref{corollary:obs}.
\end{lemma}
\begin{proof}
Let $X = \rho GH$. Recall that for any $z, \hat{z} \in \mathbb{R}^2$,
\begin{equation} \label{eq:lmi0}
\| \rho G \sin(Hz) - \rho G \sin(H \hat{z}) \| \leq  \| X (z - \hat{z}) \|.
\end{equation}
Let $w = \rho G \sin(H z) - \rho G \sin(H \hat{z})$ and $\Delta z = z - \hat{z}$, \eqref{eq:lmi0} is the same as $w^\top w - \Delta z^\top X^\top X  \Delta z \leq 0$, that is,
\begin{equation} \label{eq:lmi1}
\begin{bmatrix}
\Delta z^\top & w^\top 
\end{bmatrix} 
\begin{bmatrix}
- X^\top X & 0 \\ 0 & I
\end{bmatrix} 
\begin{bmatrix}
\Delta z \\ w
\end{bmatrix} \leq 0.
\end{equation}
Let $A_o = A - P^{-1} ZC$. Eq.~\eqref{eq:coro-uioc-obs} holds if there exists $\theta \in (0, 1)$ such that $(A_o \Delta z + w)^\top P (A_o \Delta z+w) - \theta \Delta z^\top P \Delta z \leq 0$ \cite[Theorem~1.4]{bof2018lyapunov}, i.e.,
\begin{equation}\label{eq:lmi2}
\begin{bmatrix}
\Delta z^\top & w^\top 
\end{bmatrix}
\begin{bmatrix}
A_o^\top P A_o - \theta P  & A_o^\top P \\ PA_o &  P
\end{bmatrix}
\begin{bmatrix}
\Delta z \\ w
\end{bmatrix} \leq 0.
\end{equation}
From \eqref{eq:lmi1}, \eqref{eq:lmi2} holds if there exists $\epsilon > 0$ such that
\begin{equation} \label{eq:lmi3}
\begin{bmatrix}
A_o^\top P A_o - \theta P  & A_o^\top P \\ PA_o &  P
\end{bmatrix} - 
\epsilon \begin{bmatrix}
- X^\top X & 0 \\ 0 & I
\end{bmatrix}  \leq 0.
\end{equation}
Assuming that $P - \epsilon I < 0$ and $\epsilon >0$, we claim that \eqref{eq:lmi3} is equivalent to
\begin{equation} \label{eq:lmi4}
\begin{bmatrix}
-\theta P & A_o^\top P & \epsilon X^\top & A_o^\top P \\
PA_o & -P & 0 & 0 \\
\epsilon X & 0 & -\epsilon I & 0 \\
PA_o & 0 & 0 & P - \epsilon I
\end{bmatrix}  \leq 0.
\end{equation}
To see this, under the assumptions $P - \epsilon I < 0$ and $\epsilon > 0$, the Schur complement shows that \eqref{eq:lmi3} is equivalent to $A_o^\top P A_o - \theta P  + \epsilon X^\top X  - A_o^\top P ( P - \epsilon I)^{-1} PA_o \leq 0$. The same condition is obtained by applying the Schur complement twice to \eqref{eq:lmi4}.
\end{proof}

\section{}
\label{app:lemma-holder}
\begin{lemma} \label{lemma:holder-like}
For any $n \times n$ Hermitian matrices $A$ and $B$, we have $|{\rm Tr}(AB)| \leq \sigma_{\max}(A) \|B\|_1$.
\end{lemma}
\begin{proof}
Let $B = \sum_{j=1}^{n} \lambda_j v_j v^\dagger_j$ be the spectral decomposition of $B$, where $\{ v_j \}$ forms an orthonormal basis for $\mathbb{R}^{n}$, where $\dagger$ is the adjoint. Let $\{e_j\}$ be the standard basis for $\mathbb{R}^{n}$. Then there exists a unitary matrix $U$ such that $v_j = U e_j$ for $1 \leq j \leq n$. Therefore,
$$\left|{\rm Tr}(AB) \right| = \left| \sum_{j=1}^{n} \lambda_j {\rm Tr} \left(e^\dagger_j U^\dagger A U e_j \right) \right| \leq \sum_{j=1}^{n} |\lambda_j| \left| (U^\dagger A U)_{jj} \right|,$$
where $X_{jj}$ is the $(j,j)$-th element of a matrix $X$. For any Hermitian matrix $A$, by the min-max theorem, we have 
$\lambda_{\min}(A) = \min_{\sqrt{x^\top x}=1} x^\top A x \leq A_{jj} \leq \lambda_{\max}(A) = \max_{\sqrt{x^\top x}=1} x^\top A x.$
Therefore, $|A_{jj}| \leq \max\{| \lambda_{\min}(A)|, |\lambda_{\max}(A)|\} = \sigma_{\max}(A)$. By unitary invariance of singular values, $\left| {\rm Tr}(AB) \right| \leq \sigma_{\max}(A) \sum_{j=1}^{n}|\lambda_j|  = \sigma_{\max}(A) \| B \|_1$.
\end{proof}

\addtolength{\textheight}{-18.2cm}

\bibliographystyle{IEEEtran}
\bibliography{sg_technical}

\end{document}